\documentclass[journal]{IEEEtran}

\usepackage[utf8]{inputenc}

\usepackage{amsmath}
\usepackage{amssymb}
\usepackage{subfigure}
\usepackage{graphicx}
\usepackage{graphics}
\usepackage{color}
\usepackage{psfrag}
\usepackage{cite}
\usepackage{balance}
\usepackage{algorithm}
\usepackage{accents}
\usepackage{amsthm}
\usepackage{bm}
\usepackage{url}
\usepackage{algorithmic}
\usepackage[english]{babel}
\usepackage{multirow}
\usepackage{enumerate}
\usepackage{cases}
\usepackage{stfloats}
\usepackage{dsfont}
\usepackage{soul}
\usepackage{amsfonts}
\usepackage{fancyhdr}
\usepackage{hhline}
\usepackage{array}
\usepackage{booktabs}
\usepackage{framed}
\usepackage{float}
\usepackage{soul}

\begin{document}

\newtheorem{theorem}{Theorem}
\newtheorem{acknowledgement}[theorem]{Acknowledgement}
\newtheorem{axiom}[theorem]{Axiom}
\newtheorem{case}[theorem]{Case}
\newtheorem{claim}[theorem]{Claim}
\newtheorem{conclusion}[theorem]{Conclusion}
\newtheorem{condition}[theorem]{Condition}
\newtheorem{conjecture}[theorem]{Conjecture}
\newtheorem{criterion}[theorem]{Criterion}
\newtheorem{definition}{Definition}
\newtheorem{exercise}[theorem]{Exercise}
\newtheorem{lemma}{Lemma}
\newtheorem{corollary}{Corollary}
\newtheorem{notation}[theorem]{Notation}
\newtheorem{problem}[theorem]{Problem}
\newtheorem{proposition}{Proposition}
\newtheorem{solution}[theorem]{Solution}
\newtheorem{summary}[theorem]{Summary}
\newtheorem{assumption}{Assumption}
\newtheorem{example}{\bf Example}
\newtheorem{remark}{\bf Remark}

\newtheorem{thm}{Corollary}[section]
\renewcommand{\thethm}{\arabic{section}.\arabic{thm}}

\def\qed{$\Box$}
\def\QED{\mbox{\phantom{m}}\nolinebreak\hfill$\,\Box$}
\def\proof{\noindent{\emph{Proof:} }}
\def\poof{\noindent{\emph{Sketch of Proof:} }}
\def
\endproof{\hspace*{\fill}~\qed
\par
\endtrivlist\unskip}
\def\endproof{\hspace*{\fill}~\qed\par\endtrivlist\vskip3pt}

\def\E{\mathsf{E}}
\def\eps{\varepsilon}
\def\phi{\varphi}
\def\Lsp{{\boldsymbol L}}
\def\Bsp{{\boldsymbol B}}
\def\lsp{{\boldsymbol\ell}}
\def\Ltsp{{\Lsp^2}}
\def\Lpsp{{\Lsp^p}}
\def\Linsp{{\Lsp^{\infty}}}
\def\LtR{{\Lsp^2(\Rst)}}
\def\ltZ{{\lsp^2(\Zst)}}
\def\ltsp{{\lsp^2}}
\def\ltZt{{\lsp^2(\Zst^{2})}}
\def\ninN{{n{\in}\Nst}}
\def\oh{{\frac{1}{2}}}
\def\grass{{\cal G}}
\def\ord{{\cal O}}
\def\dist{{d_G}}
\def\conj#1{{\overline#1}}
\def\ntoinf{{n \rightarrow \infty}}
\def\toinf{{\rightarrow \infty}}
\def\tozero{{\rightarrow 0}}
\def\trace{{\operatorname{trace}}}
\def\ord{{\cal O}}
\def\UU{{\cal U}}
\def\rank{{\operatorname{rank}}}
\def\acos{{\operatorname{acos}}}

\def\SINR{\mathsf{SINR}}
\def\SNR{\mathsf{SNR}}
\def\SIR{\mathsf{SIR}}
\def\tSIR{\widetilde{\mathsf{SIR}}}
\def\Ei{\mathsf{Ei}}
\def\l{\left}
\def\r{\right}
\def\lb{\left\{}
\def\rb{\right\}}

\setcounter{page}{1}

\newcommand{\eref}[1]{(\ref{#1})}
\newcommand{\fig}[1]{Fig.\ \ref{#1}}

\def\bydef{:=}
\def\ba{{\mathbf{a}}}
\def\bb{{\mathbf{b}}}
\def\bc{{\mathbf{c}}}
\def\bd{{\mathbf{d}}}
\def\bee{{\mathbf{e}}}
\def\bff{{\mathbf{f}}}
\def\bg{{\mathbf{g}}}
\def\bh{{\mathbf{h}}}
\def\bi{{\mathbf{i}}}
\def\bj{{\mathbf{j}}}
\def\bk{{\mathbf{k}}}
\def\bl{{\mathbf{l}}}
\def\bm{{\mathbf{m}}}
\def\bn{{\mathbf{n}}}
\def\bo{{\mathbf{o}}}
\def\bp{{\mathbf{p}}}
\def\bq{{\mathbf{q}}}
\def\br{{\mathbf{r}}}
\def\bs{{\mathbf{s}}}
\def\bt{{\mathbf{t}}}
\def\bu{{\mathbf{u}}}
\def\bv{{\mathbf{v}}}
\def\bw{{\mathbf{w}}}
\def\bx{{\mathbf{x}}}
\def\by{{\mathbf{y}}}
\def\bz{{\mathbf{z}}}
\def\b0{{\mathbf{0}}}

\def\bA{{\mathbf{A}}}
\def\bB{{\mathbf{B}}}
\def\bC{{\mathbf{C}}}
\def\bD{{\mathbf{D}}}
\def\bE{{\mathbf{E}}}
\def\bF{{\mathbf{F}}}
\def\bG{{\mathbf{G}}}
\def\bH{{\mathbf{H}}}
\def\bI{{\mathbf{I}}}
\def\bJ{{\mathbf{J}}}
\def\bK{{\mathbf{K}}}
\def\bL{{\mathbf{L}}}
\def\bM{{\mathbf{M}}}
\def\bN{{\mathbf{N}}}
\def\bO{{\mathbf{O}}}
\def\bP{{\mathbf{P}}}
\def\bQ{{\mathbf{Q}}}
\def\bR{{\mathbf{R}}}
\def\bS{{\mathbf{S}}}
\def\bT{{\mathbf{T}}}
\def\bU{{\mathbf{U}}}
\def\bV{{\mathbf{V}}}
\def\bW{{\mathbf{W}}}
\def\bX{{\mathbf{X}}}
\def\bY{{\mathbf{Y}}}
\def\bZ{{\mathbf{Z}}}

\def\bxi{{\boldsymbol{\xi}}}
\def\bmu{{\boldsymbol{\mu}}}
\def\bnu{{\boldsymbol{\nu}}}
\def\bSigma{{\boldsymbol{\Sigma}}}
\def\ep{\overline{\varepsilon_p}}

\def\sT{{\mathsf{T}}}
\def\sH{{\mathsf{H}}}

\def\d{\mathrm{d}}

\def\tcK{{\widetilde{\mathcal{K}}}}

\def\tT{{\widetilde{T}}}
\def\tF{{\widetilde{F}}}
\def\tP{{\widetilde{P}}}
\def\tG{{\widetilde{G}}}
\def\tbh{{\widetilde{\mathbf{h}}}}
\def\tbg{{\widetilde{\mathbf{g}}}}
\def\tbx{{\tilde{\mathbf{x}}}}
\def\tx{{\tilde{x}}}
\def\hbx{{\hat{\mathbf{x}}}}

\def\mA{{\mathbb{A}}}
\def\mB{{\mathbb{B}}}
\def\mC{{\mathbb{C}}}
\def\mD{{\mathbb{D}}}
\def\mE{{\mathbb{E}}}
\def\mF{{\mathbb{F}}}
\def\mG{{\mathbb{G}}}
\def\mH{{\mathbb{H}}}
\def\mI{{\mathbb{I}}}
\def\mJ{{\mathbb{J}}}
\def\mK{{\mathbb{K}}}
\def\mL{{\mathbb{L}}}
\def\mM{{\mathbb{M}}}
\def\mN{{\mathbb{N}}}
\def\mO{{\mathbb{O}}}
\def\mP{{\mathbb{P}}}
\def\mQ{{\mathbb{Q}}}
\def\mR{{\mathbb{R}}}
\def\mS{{\mathbb{S}}}
\def\mT{{\mathbb{T}}}
\def\mU{{\mathbb{U}}}
\def\mV{{\mathbb{V}}}
\def\mW{{\mathbb{W}}}
\def\mX{{\mathbb{X}}}
\def\mY{{\mathbb{Y}}}
\def\mZ{{\mathbb{Z}}}

\def\cA{\mathcal{A}}
\def\cB{\mathcal{B}}
\def\cC{\mathcal{C}}
\def\cD{\mathcal{D}}
\def\cE{\mathcal{E}}
\def\cF{\mathcal{F}}
\def\cG{\mathcal{G}}
\def\cH{\mathcal{H}}
\def\cI{\mathcal{I}}
\def\cJ{\mathcal{J}}
\def\cK{\mathcal{K}}
\def\cL{\mathcal{L}}
\def\cM{\mathcal{M}}
\def\cN{\mathcal{N}}
\def\cO{\mathcal{O}}
\def\cP{\mathcal{P}}
\def\cQ{\mathcal{Q}}
\def\cR{\mathcal{R}}
\def\cS{\mathcal{S}}
\def\cT{\mathcal{T}}
\def\cU{\mathcal{U}}
\def\cV{\mathcal{V}}
\def\cW{\mathcal{W}}
\def\cX{\mathcal{X}}
\def\cY{\mathcal{Y}}
\def\cZ{\mathcal{Z}}
\def\cd{\mathcal{d}}
\def\Mt{M_{t}}
\def\Mr{M_{r}}
\newcommand{\figref}[1]{{Fig.}~\ref{#1}}
\newcommand{\tabref}[1]{{Table}~\ref{#1}}

\newcommand{\cov}{\mathsf{Cov}}
\newcommand{\var}{\mathsf{Var}}
\newcommand{\fb}{\tx{fb}}
\newcommand{\nf}{\tx{nf}}
\newcommand{\BC}{\tx{(bc)}}
\newcommand{\MAC}{\tx{(mac)}}
\newcommand{\Pout}{p_{\mathsf{out}}}
\newcommand{\nnn}{\nn\\}
\newcommand{\FB}{\tx{FB}}
\newcommand{\TX}{\tx{TX}}
\newcommand{\RX}{\tx{RX}}
\renewcommand{\mod}{\tx{mod}}
\newcommand{\m}[1]{\mathbf{#1}}
\newcommand{\td}[1]{\tilde{#1}}
\newcommand{\sbf}[1]{\scriptsize{\textbf{#1}}}
\newcommand{\stxt}[1]{\scriptsize{\textrm{#1}}}
\newcommand{\suml}[2]{\sum\limits_{#1}^{#2}}
\newcommand{\sumlk}{\sum\limits_{k=0}^{K-1}}
\newcommand{\eqhsp}{\hspace{10 pt}}
\newcommand{\Hz}{\ \tx{Hz}}
\newcommand{\sinc}{\tx{sinc}}
\newcommand{\tr}{\mathsf{tr}}
\newcommand{\diag}{\mathrm{diag}}
\newcommand{\MAI}{\tx{MAI}}
\newcommand{\ISI}{\tx{ISI}}
\newcommand{\IBI}{\tx{IBI}}
\newcommand{\CN}{\tx{CN}}
\newcommand{\CP}{\tx{CP}}
\newcommand{\ZP}{\tx{ZP}}
\newcommand{\ZF}{\tx{ZF}}
\newcommand{\SP}{\tx{SP}}
\newcommand{\MMSE}{\tx{MMSE}}
\newcommand{\MINF}{\tx{MINF}}
\newcommand{\RC}{\tx{MP}}
\newcommand{\MBER}{\tx{MBER}}
\newcommand{\MSNR}{\tx{MSNR}}
\newcommand{\MCAP}{\tx{MCAP}}
\newcommand{\vol}{\tx{vol}}
\newcommand{\ah}{\hat{g}}
\newcommand{\tg}{\tilde{g}}
\newcommand{\teta}{\tilde{\eta}}
\newcommand{\heta}{\hat{\eta}}
\newcommand{\uh}{\m{\hat{s}}}
\newcommand{\eh}{\m{\hat{\eta}}}
\newcommand{\hv}{\m{h}}
\newcommand{\hh}{\m{\hat{h}}}
\newcommand{\Po}{P_{\mathrm{out}}}
\newcommand{\Poh}{\hat{P}_{\mathrm{out}}}
\newcommand{\Ph}{\hat{\gamma}}
\newcommand{\mat}[1]{\begin{matrix}#1\end{matrix}}
\newcommand{\ud}{^{\dagger}}
\newcommand{\C}{\mathcal{C}}
\newcommand{\nn}{\nonumber}
\newcommand{\nInf}{U\rightarrow \infty}

\title{Ultra-Low-Latency Edge Intelligent Sensing: A Source-Channel Tradeoff and Its Application to Coding Rate Adaptation}
\author{Qunsong Zeng,~\IEEEmembership{Member,~IEEE}, Jianhao Huang,~\IEEEmembership{Member,~IEEE}, Zhanwei Wang,~\IEEEmembership{Graduate Student Member,~IEEE}, Kaibin Huang,~\IEEEmembership{Fellow,~IEEE}, and Kin K. Leung,~\IEEEmembership{Fellow,~IEEE}
\thanks{Q. Zeng, J. Huang, Z. Wang, and K. Huang are with Department of Electrical \& Electronic Engineering, The University of Hong Kong, Hong Kong, China. K. K. Leung is with both Electrical \& Electronic Engineering and Computing Departments, Imperial College London, London, UK.

\emph{Corresponding author: Kaibin Huang} (Email: huangkb@eee.hku.hk).}\vspace{-10mm}
}

\maketitle

\begin{abstract}
The forthcoming sixth-generation (6G) mobile network is set to merge edge artificial intelligence (AI) and integrated sensing and communication (ISAC) extensively, giving rise to the new paradigm of \emph{edge intelligent sensing} (EI-Sense). 
This paradigm leverages ubiquitous edge devices for environmental sensing and deploys AI algorithms at edge servers to interpret the observations via remote inference on wirelessly uploaded features. 
A significant challenge arises in designing EI-Sense systems for 6G mission-critical applications, which demand high performance under stringent latency constraints. 
To tackle this challenge, we focus on the end-to-end (E2E) performance of EI-Sense and characterize a source-channel tradeoff that balances source distortion and channel reliability.
In this work, we establish a theoretical foundation for the source-channel tradeoff by quantifying the effects of source coding on feature discriminant gains and channel reliability on packet loss. 
Building on this foundation, we design the coding rate control by optimizing the tradeoff to minimize the E2E sensing error probability, leading to a low-complexity algorithm for ultra-low-latency EI-Sense. 
Finally, we validate our theoretical analysis and proposed coding rate control algorithm through extensive experiments on both synthetic and real datasets, demonstrating the sensing performance gain of our approach with respect to traditional reliability-centric methods.

\end{abstract}
\begin{IEEEkeywords}
Sensing, edge AI, quantization, short packet transmission, ultra-low-latency sensing and communication.
\end{IEEEkeywords}    
\vspace{-4mm}
\section{Introduction}
Beyond a scaled-up version of 5G, the upcoming 6G mobile networks introduce two new usage scenarios: \emph{integrated sensing and communication} (ISAC) and \emph{integrated AI and communication} (IAAC)~\cite{ITU-R_M2160_2023}. The former exploits the ubiquitous edge devices as sensors for environment perception~cite{cui2021integrating}. The latter involves widespread deployment of AI algorithms at the network edge for providing intelligence services to edge devices, giving rise to the area of edge AI~\cite{letaief2019roadmap}. Their natural fusion, termed \emph{edge intelligent sensing} (EI-Sense), promises to offer a unified platform for supporting a wide range of applications, e.g., industrial automation, autonomous driving, and robotic control~\cite{liu2025integrated}. Among these, 6G mission-critical applications pose a significant challenge for EI-Sense design, as they demand high performance under stringent latency constraints~\cite{letaief2022edge}. To tackle this challenge, we consider an ultra-low-latency EI-Sense system, where the remote inference at an edge server leverages observations uploaded by sensors by short packets. In this work, we study the end-to-end (E2E) performance of EI-Sense by mathematically characterizing a tradeoff between source distortion and channel reliability, termed the source-channel tradeoff. By designing a coding rate adaptation scheme, this tradeoff is optimized for maximizing the sensing performance under strict latency and radio resource constraints.

The advancements of EI-Sense can leverage those in the areas of ISAC and edge AI. The former has primarily focused on developing dual-functional systems and techniques to efficiently utilize shared radio resources, thereby improving the spectral efficiency and reducing latency~\cite{liu2022integrated,liu2020jointradar,zhang2021enabling}. In practice, these benefits have been exploited to support autonomous driving with enhanced road safety~\cite{kumari2020adaptive} and industrial automation through real-time monitoring and control of manufacturing processes~\cite{zhu2024enabling}. Existing research in ISAC has already explored the integration with AI by typically applying machine learning algorithms as optimization tools to address challenges in sensing or communication, such as adaptive resource allocation~\cite{liu2024optimal}, interference mitigation~\cite{qi2024deep}, and joint optimization of sensing and communication functions ~\cite{liu2022learning}. However, a systematic fusion approach should involve positioning edge AI as the platform of providing remote intelligence to execute tasks that rely on observations from sensors~\cite{chen2024view}. This EI-Sense backbone architecture is closely related to a common approach in edge AI known as split inference~\cite{wang2020convergence}. In split inference, sensors use a lightweight neural network model to extract features from local sensing data, which are then uploaded wirelessly to the edge server for inference using a pre-trained deep neural network model~\cite{wu2024device,li2020edge,wen2024taskoriented}. Though the compression of raw data into more informative features reduces overhead significantly, the communication bottleneck still exists due to the high dimensionality of features~\cite{xu2018scaling}. One solution to overcome the bottleneck is to prune features based on their importance levels for accurate inference~\cite{shao2020communication, guo2020channel}. However, this required joint adjustment of both sensor and server models limits the approach’s flexibility and adaptability. Another popular solution, known as joint source-and-channel coding, employs an auto-encoder pair consisting of an encoder and a decoder deployed separately at the two ends of the channel. These are jointly trained to perform efficient feature extraction, inference, and coping with channel noise simultaneously~\cite{jankowski2020wireless, wu2024deep}. However, this black-box method lacks interpretability and hence is difficult to derive useful insights into the system’s optimal operations, thereby limiting its compatibility with existing digital techniques for coding, modulation, and adaptive transmissions. At a higher level, most existing wireless techniques designed for EI-Sense-relevant systems rely on the assumption of long-packet transmission. Thereby, they fall short on providing solutions for 6G ultra-low-latency applications.
 
The key 5G paradigm of ultra-reliable low-latency communication (URLLC) will be further advanced with the realization of the 6G vision of hyper-reliable low-latency communication (HRLLC)~\cite{hong20226g}. Specifically, air latency will be reduced to as low as 0.1 milliseconds and reliability dramatically increased to attain an extremely low packet error rate of $10^{-7}$~\cite{tao20236g}. To achieve the ultra-low-latency targets of URLLC or HRLLC relies on short-packet transmission, which ensures real-time data delivery~\cite{durisi2016toward, popovski2019wireless}. Underpinning the relevant designs is the finite block-length information theory, which accounts for the fact that the decoding error rate of short packets can no longer be negligible and is a function of packet length~\cite{polyanskiy2010channel,ersegh2015evaluation}. Due to the inherent tradeoffs among data rate, latency, and reliability, existing URLLC/HRLLC techniques have to operate at low data rates in return for low-latency and high reliability, limiting their applications to those requiring only short messages, e.g., commands or emergency alerts~\cite{wu2024goal}. They are inadequate for supporting AI-enabled applications in 6G as they are data intensive. To boost the data rates of URLLC has motivated researchers to streamline relevant techniques such as resource allocation~\cite{Nasir2021resource}, frame structure design~\cite{zhou2023joint}, and joint power-and-blocklength optimization~\cite{ren2019joint}. Despite such efforts, meeting the performance requirements for ultra-low-latency EI-Sense still needs major breakthroughs in supporting real-time transmission of high-dimensional features. A critical issue arises when assembling large-bit feature vectors into short packets, which induces a high coding rate and can lead to severe packet loss. This necessitates the use of aggressive source coding to compress the feature vectors into fewer bits but doing so would incur severe source distortion and thereby degrade the sensing performance. Therefore, it is essential to carefully balance the source-channel tradeoff for developing URLLC/HRLLC techniques to support EI-Sense applications, which motivates this work.

In this work, we present a novel ultra-low-latency EI-Sense framework. As a key component, the source-channel tradeoff is derived to establish a theoretical foundation for designing the EI-Sense system and techniques. Specifically, the tractable analysis of this tradeoff provides a theoretic tool for optimizing the EI-Sense performance, yielding techniques for reliable feature transmission under stringent constraints on latency and radio resources. In the considered system with sequential sensing, a single-antenna sensor captures a sequence of observations and uploads the quantized features to a multi-antenna edge server for remote inference. Multiple sensor observations help to improve the sensing performance via view diversity~\cite{zeng2024knowledge}. Under extremely low latency constraints, acquiring channel state information at the transmitter side (CSIT) becomes extremely costly and impractical~\cite{lopez2020ultralow}. Assuming no CSIT, we consider fixed-power transmission while the increased error rate rising from lack of channel adaptation would be absorbed by exploiting the robustness of inference model. Furthermore, the E2E performance is enhanced by balancing the source-channel tradeoff through code rate adaptation. Our tractable analysis is based on the assumption that the feature vectors follow a Gaussian mixture model, which is widely used in statistical inference~\cite{hastie2009elements} and deep learning~\cite{ye2014realtime,figueroa2019semi}. The resultant results are subsequently validated using real datasets. The main contributions and key findings of this work are summarized as follows.
\begin{itemize}
    \item \textbf{Source-channel tradeoff analysis:} The tradeoff builds on mathematical analysis the effects of source coding and channel reliability on E2E sensing performance. On one hand, the effect of source coding (i.e., block quantization) is reflected in the reduction on discriminant gain of features extracted from sensing data. Using a practical uniform scalar quantizer with transform coding, we show that the difference between the reduced and original (without quantisation) discriminant gains is on the order of $O(4^{-R})$, where $R$ represents the number of bits used to quantize a single feature. On the other hand, channel reliability, measured by the packet loss probability, affects the number of feature vectors successfully received. By invoking the finite-blocklength information theory, a higher coding rate increases the packet loss probability or equivalently reduces the number of received feature vectors. By combining the preceding source-channel analysis, the derived source-channel tradeoff reveals relations among discriminant gain, packet loss probability, and the number of observations, as well as their collective effects on sensing performance. In addition, we quantify the influence of the number of receive antennas at the edge server and discuss the relationship between channel diversity and feature diversity.
    \item \textbf{Coding rate optimization:} The coding rate is optimised as a variable in the preceding source-channel tradeoff with the objective of minimizing the E2E sensing error probability. The difficulty in solving this non-convex problem is overcome by proposing a reasonable surrogate function of coding rate that is shown to be concave and facilitate finding the optimal coding rate using an existing convex optimization toolbox. Furthermore, we derive a closed-form approximation of the surrogate function’s gradient. The result enables the development of a low-complexity optimization algorithm based on the gradient ascent method.
    \item \textbf{Experiments:} The analytical results and the effectiveness of proposed algorithms are validated through experiments on both synthetic and real datasets. Especially, leveraging insights from our analytical results, we extend our approach to the case of convolutional neural network (CNN) as the classifier on a real multi-view dataset. The proposed adaptive coding-rate design demonstrates near-optimal performance and consistently outperforms traditional reliability-centric URLLC techniques in terms of E2E sensing performance. 
\end{itemize}

The remainder of the paper is organized as follows. 
The models and metrics are elaborated in Section \ref{Sec: System model}, followed by preliminaries in Section~\ref{Sec: Preliminary and system configuration}.
Section \ref{Sec: Analysis} analyzes the impacts of source distortion and channel reliability on sensing performance, and discusses the source-channel tradeoff. Coding rate optimization is covered in Section \ref{Sec: Tradeoff and optimization}.
Experimental results are provided in Section \ref{Sec: Experimental Results}. The paper is concluded in Section \ref{Sec: Conclusion}.

\begin{figure*}[t!]
    \centering
    \subfigure[E2E EI-Sense system]{
    \includegraphics[width=1.86\columnwidth]{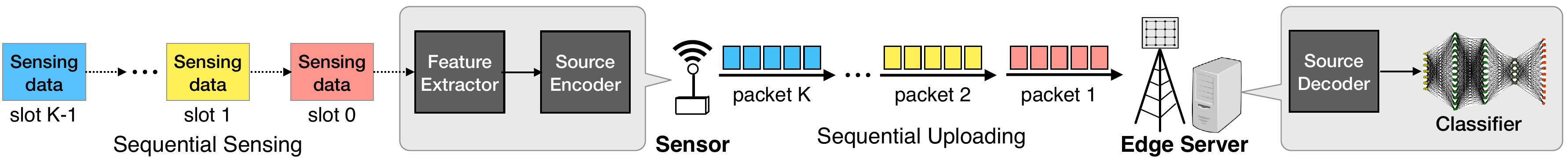}}
    \subfigure[Sequential sensing operation]{
    \includegraphics[width=1.125\columnwidth]{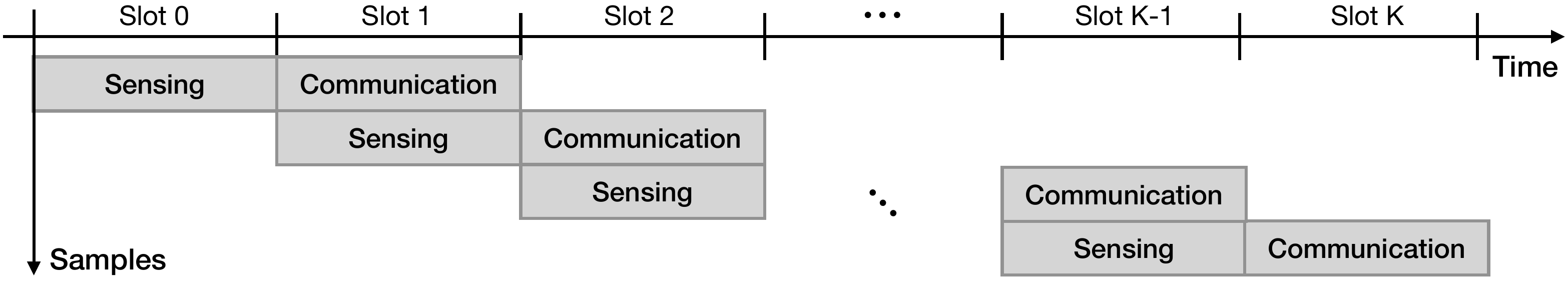}}
    \caption{An illustration of an E2E ultra-low-latency EI-Sense system. (a) System architecture and operations at the sensor and edge server. (b) Parallelization between sensing and communication in the sequential sensing scenario.\vspace{-8mm}}\label{Fig: E2E system}
\end{figure*}

\section{Models and Metrics}\label{Sec: System model}

We consider an E2E ultra-low-latency EI-Sense system as shown in Fig.~\ref{Fig: E2E system}, where $K$ local observations from the sensor are sequentially perceived, compressed into feature vectors, and sent to the edge server for remote inference.
Relevant models and metrics are described as follows.
\vspace{-3mm}
\subsection{Data and Inference Models}
The sensor produces $K$ observations of a target object, labeled as class $\ell$, which is assumed to be drawn uniformly from a set of $L_c$ classes.
Specifically, for the $k$-th observation, the sensor generates a feature vector that consists of $d$ measurements, represented as $\bx_k\in\mR^{d}$. 
These feature vectors are assumed to be randomly sampled from a joint conditional distribution $(\bx_1,\cdots,\bx_K)\sim p_{(\bx_1,\cdots,\bx_K)}(\bx_1,\cdots,\bx_K|\ell)$.
As such, the joint feature vector distribution can be expressed as
\begin{equation}
    p(\bx_1,\cdots,\bx_K)\sim\frac{1}{L_c}\sum_{\ell=1}^{L_c} p_{(\bx_1,\cdots,\bx_K)}(\bx_1,\cdots,\bx_K|\ell).
\end{equation}
Further explications of the distribution are provided for the two types of classifier models, respectively.
\subsubsection{Gaussian Mixture Model for Statistical Inference}
For reasons of tractability and practicality, we focus on the scenario where each feature vector originates from a Gaussian mixture distribution~\cite{figueroa2019semi,lan2023progressive}. 
Specifically, when conditioned on the $\ell$-th class, each feature vector $\bx_k$ observed by the sensor is independently drawn from a multivariate Gaussian distribution. 
This distribution has a mean of $\bmu_{\ell}\in\mR^{d}$ and covariance matrix $\bSigma\in\mR^{d\times d}$, giving that $(\bx_k|\ell)\sim\cN(\bmu_{\ell},\bSigma)$. 
Then, the \emph{probability density function} (PDF) of the joint feature vector distribution for $(\bx_1,\cdots,\bx_K)$ is represented as 
\begin{equation}\label{Eqn: joint GMM}
    p(\bx_1,\cdots,\bx_K)=\frac{1}{L_c}\sum_{\ell=1}^{L_c}\prod_{k=1}^K\cN(\bx_k|\bmu_{\ell},\bSigma),
\end{equation}
where $\cN(\bx_k|\bmu_{\ell},\bSigma)$ denotes the PDF of the multivariate Gaussian distribution with mean $\bmu_{\ell}$ and covariance matrix $\bSigma$. 
We consider the \emph{maximum a posterior} (MAP) classifier for the distribution in \eqref{Eqn: joint GMM}. 
Due to the uniform prior probabilities of classes, the MAP classifier is equivalent to a \emph{maximum likelihood} classifier, which gives the estimated label $\hat{\ell}$ as 
\begin{align}
    \hat{\ell}=\arg\max_{\ell}\Pr(\bx_1,\cdots,\bx_K|\ell)=\arg\max_{\ell}\prod_{k=1}^Kp_{\ell}(\bx_k),
\end{align}
where $p_{\ell}(\bx_k)\triangleq\Pr(\bx_k|\ell)$ is the likelihood function of the $k$-th observation by the sensor.
\subsubsection{General Model for CNN Classification}
We also consider the practical CNN classifier model to validate insights from analysis. 
For CNN classifiers, the architecture comprises multiple convolutional (CONV) layers followed by multiple fully-connected layers and a softmax output activation function that outputs a confidence score of each label. 
To implement the split inference, the classifier is partitioned into the sensor and server sub-models, represented as functions $f_{\sf sensor}(\cdot)$ and $f_{\sf server}(\cdot)$, respectively. 
Let $N_c$ denote the number of CONV filters in the final layer of $f_{\sf sensor}(\cdot)$, and each filter generates a feature map with dimensions $L_h$ and $L_w$. 
The set of all extracted feature maps from an input image is represented by the tensor $\bX\in\mR^{N_c\times L_h\times L_w}$. 
The feature maps are then aggregated and the confidence score of the server-side classifier can be obtained by feeding the feature maps into the server sub-model, i.e., $\{s_1,\cdots,s_{L_c}\}=f_{\sf server}(\bX_1,\cdots,\bX_K)$. 
The classifier outputs the inferred label that has the maximum confidence score, giving that $\hat{\ell}=\arg\max_{\ell}s_{\ell}$.

\vspace{-4mm}
\subsection{Source Coding (Block Quantization)}
To ease the notation, we focus on the source coding of the feature vector $\bx_k$ from the $k$-th observation and omit the subscript $k$ without causing any confusion.
The purpose of source coding is to compress data into bit streams within an acceptable level of distortion. 
This process involves an encoder $\cF(\cdot)$ at the sensor side and a decoder $\cF^{-1}(\cdot)$ at the edge server side, respectively.

The sensor-side encoder follows a standard procedure for quantizing the blocks of $d$ correlated Gaussian random variables~\cite{huang1963block}. This procedure typically includes transform coding and quantization for data compression. 
To achieve superior quantization, we apply the transform coding, such as the Karhunen-Lo\`eve transform which is the optimal transform for Gaussian sources~\cite{akyol2009transform}. 
The orthogonal transform yields that $\tbx=\bA\bx$, where the transform matrix $\bA$ satisfies $\bA^{\sT}\bA=\bI$, making the coefficients in $\tbx$ uncorrelated. 
Thus, we can conduct subsequent element-wise scalar quantization. 
We assume that the transformed feature vector $\tbx$ can only take values in a finite constellation set $\cS\subset\mR^{d}$. 
The cardinality constraint on $\cS$ due to quantization is expressed as $\log_2|\cS|\leq Rd$, where $R$ denotes the number of bits per dimension. 
For simplicity and ease of implementation~\cite{saha2022decentralized}, we focus on the uniform scalar quantizers to illustrate system design. 
Along each dimension, we have a total of $2^R$ quantization points which we denote here as $\{u_1,\cdots,u_{2^R}\}$, where $u_1=-U$ and $u_{2^R}=+U$. 
Then, the quantization resolution is given as
\begin{equation}
    \Delta=\frac{2U}{2^R-1}.
\end{equation}
We denote the block quantizer as $\cQ_{\Delta,U}(\cdot)$, where the element-wise uniform scalar quantization scheme is employed. 
For an input vector $\tbx=[\tx_1,\cdots,\tx_d]^{\sT}\in\mR^d$, the quantizer output is
\begin{equation}
    \tbx'=\cQ_{\Delta,U}(\tbx)=[q(\tx_1),\cdots,q(\tx_d)]^{\sT},
\end{equation}
where the operation $q(\cdot)$ is defined as $q(x)=u_i$ for $x\in[u_{i}-\Delta/2,u_{i}+\Delta/2)$, $i=1,\cdots,M$.
Consequently, the source encoder produces $\tbx'=\cF(\bx)=\cQ_{\Delta,U}(\bA\bx)$.

The server-side source decoder executes the transform decoding to retrieve the original feature vector in the following manner: $\hbx=\cF^{-1}(\tbx')=\bA^{\sT}\tbx'$.
Finally, the decoded feature vector $\hbx$ is fed into the classifier for inference.
\vspace{-3mm}
\subsection{Short-Packet Transmission}\label{Sec: Short packet transmission}
For each sensing observation, say feature vector $\bx_k$, the sensor uploads the encoded feature vector $\tbx'_k$ to the edge server over a wireless channel, with an allocated bandwidth of $B$. 
Each time slot is required to meet the stringent latency constraint of $T$. 
Hence, the number of channel uses that allowed for each sensor in one time slot is fixed as $N=TB$. 
In this context, the coding rate of the sensor is expressed as 
\begin{equation}
    R_{c}=\frac{Rd}{N}.\label{Eqn: code rate}
\end{equation} 

During transmission, we assume \emph{independent and identically distributed} (i.i.d.) block fading channels between the sensor and the edge server.
In this case, the channel coefficients remain constant during each feature vector's transmission while i.i.d. varies among different time slots\footnote{For 5G millimeter-wave communications operating in the 24--100 GHz range, consider a carrier frequency of 60 GHz. Assuming a relative velocity of 10 m/s, the resulting coherence time is approximately 0.1 ms~\cite{rappaport2002wireless}, which corresponds to a single time slot.}.
We consider the scenario that the sensor has no knowledge of its CSIT and it uses the full power for transmission in each time slot. 
This is a practical requirement for ultra-low-latency transmission, especially in the presence of SIMO systems, where the feedback of channel vector consumes radio resource and increases the latency. 
Additionally, we do not employ a retransmission strategy for packet loss to eliminate the feedback overhead required in order to meet the ultra-low-latency constraint. Moreover, continuous sensing inherently produces i.i.d. observations over time, making retransmissions unnecessary.
According to the finite blocklength information theory~\cite{polyanskiy2010channel}, the decoding error probability in transmitting the $k$-th encoded feature vector $\tbx'_k$ can be closely approximated as 
\begin{equation}
    \varepsilon_{p,k}=Q\l(\sqrt{\frac{N}{V(\gamma_k)}}\l(C(\gamma_k)-R_{c}\r)\r),\label{Eqn: packet loss probability}
\end{equation}
where $\gamma_k$ is the received SNR in the $k$-th slot, $Q(\cdot)$ denotes the Q-function defined as $Q(x)=\frac{1}{\sqrt{2\pi}}\int_{x}^{\infty}\exp\l(-\frac{t^2}{2}\r)\d t$, $C(\gamma_k)$ is the Shannon capacity known as
    $C(\gamma_k)=\log_2(1+\gamma_k),$
and $V(\gamma_k)$ represents the channel dispersion \cite{polyanskiy2009dispersion} expressed as 
    $V(\gamma_k)=\frac{\gamma_k\l(2+\gamma_k\r)}{\l(1+\gamma_k\r)^2}(\log_2e)^2.$
\vspace{-3mm}
\subsection{Communication Model}
The single-antenna sensor sequentially transmits $K$ feature vectors to an edge server equipped with $L$ antennas.
In the $k$-th time slot, the signal received by the server from the sensor through the $l$-th stream can be written as
\begin{equation}
    \by_{k}^{(l)}=h_{k}^{(l)}\bs_{k}+\bz_{k}^{(l)},~~l\in\{1,\cdots,L\},
\end{equation}
where the transmitted packet $\bs_{k}\in\mC^{N\times1}$ satisfies $\mE[\bs_{k}(\bs_{k})^{\sH}]=P_0\bI$ with $P_0$ being the signal transmit power of the sensor, $h_{k}^{(l)}\sim\cC\cN(0,1)$ represents i.i.d. complex Gaussian coefficients in the Rayleigh fading channel, and $\bz_{k}^{(l)}\sim\cC\cN(\b0,N_0\bI)$ is \emph{additive white Gaussian noise} (AWGN). 
The transmit SNR is defined as $\gamma_0={P_0}/{N_0}$.
To decode the signal, we consider the \emph{maximal ratio combining} (MRC), i.e., 
\begin{equation}
    \by_{k}=\sum_{l=1}^{L}(h_{k}^{(l)})^*\by_{k}^{(l)}=\|\bh_{k}\|^2\bs_{k}+\sum_{l=1}^{L}(h_{k}^{(l)})^*\bz_{k}^{(l)},
\end{equation}
and thus the post-processing SNR in the $k$-th time slot is given by $\gamma_{k}=\gamma_0\|\bh_{k}\|^2$. 

\vspace{-3mm}
\subsection{Performance Metric}
To characterize the inference performance, useful metrics are described as follows.
\subsubsection{Discriminant Gain}
In statistical inference using Gaussian mixture model, the discriminant gain is introduced to measure the discernibility between any pair of classes~\cite{lan2023progressive}. 
Specifically, the pair-wise discriminant gain for class-$\ell$ and class-$\ell'$ in the feature subspace spanned by observations of the sensor is defined as
\begin{equation}\label{Eqn: discriminant gain}
    D(\ell,\ell')=\frac{1}{2}(\bmu_{\ell}-\bmu_{\ell'})^{\sT}\bSigma^{-1}(\bmu_{\ell}-\bmu_{\ell'}),
\end{equation}
which represents half of the squared \emph{Mahalanobis distance} between the centroids of the two clusters in this feature space.
\subsubsection{Sensing Error Probability}
The performance of the EI-Sense system is evaluated by the sensing error probability, which is the likelihood of a sample being incorrectly classified at the edge server side. 
For statistical inference, it refers to the \emph{Bayes error}. Given a uniform prior, the conditional error given the feature vector $\bx$ is $e(\bx)=1-\max_{\ell}\Pr(\ell|\bx)$. 
The Bayes error for sensing results is the expected value of $e(\bx)$ over $\bx$, that is, $P_e=\mE[e(\bx)]$. 
For CNN case, it refers to the \emph{Top-1 error rate} which measures how frequently the classifier does not assign the top score to the correct class.
It signifies the fraction of received samples where the class, predicted with the highest score, aligns with the true label. 
Specifically, the metric calculates the proportion of samples for which the CNN classifier's prediction is incorrect.

\begin{figure}[t!]
    \centering
    \subfigure[Gaussian mixture model]{\includegraphics[width=0.48\columnwidth]{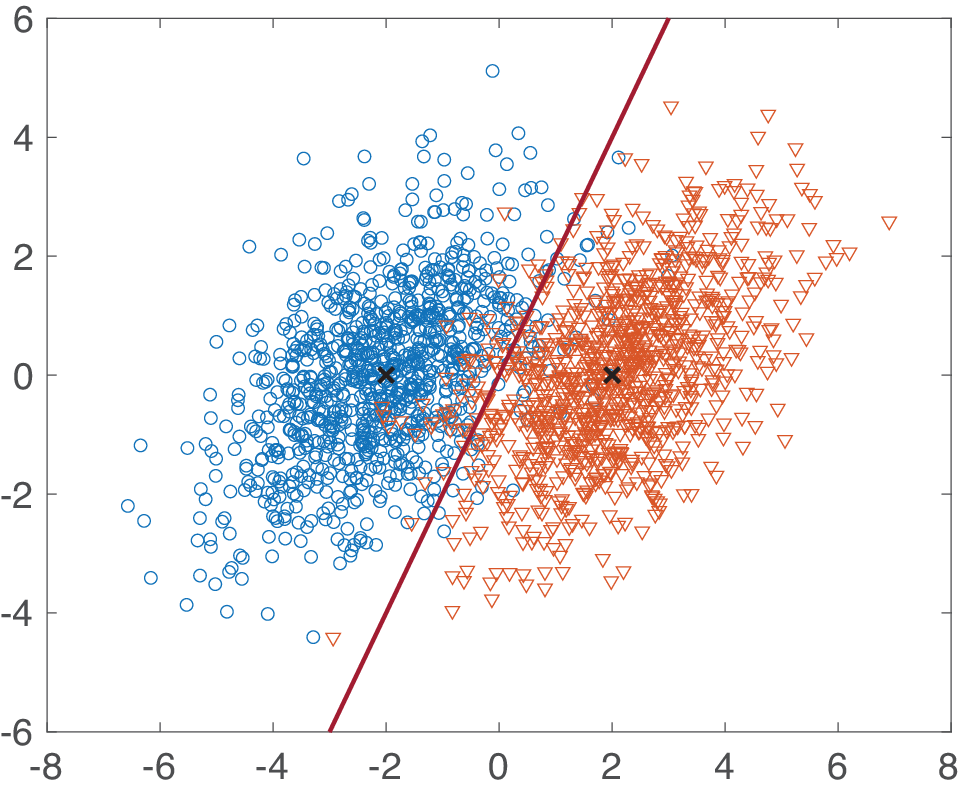}
    \label{Fig: GMM}}
    \subfigure[Discriminant score function]{\includegraphics[width=0.48\columnwidth]{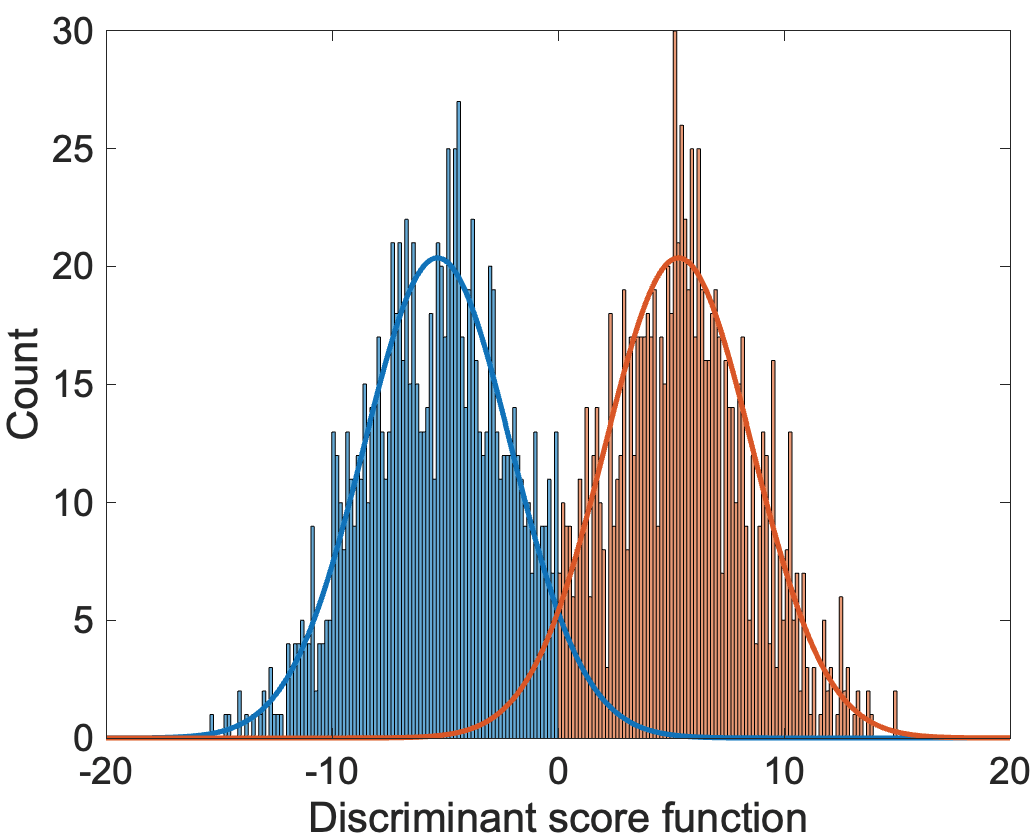}
    \label{Fig: Disciminant function}}
    \caption{An illustration of statistical inference. (a) Gaussian Mixture Model with the red line indicating the decision boundary. (b) Distribution of the discriminant score function fitted with two Gaussian distributions.\vspace{-4mm}}
\end{figure}

\vspace{-3mm}
\section{Preliminary of Statistical Inference}\label{Sec: Preliminary and system configuration}
For the purpose of analytical tractability, we consider statistical inference being adopted at the edge server. The derived results and design are validated for the case of deep neural networks by experiments in Section~\ref{Sec: CNN Classification Case}. In this section, the preliminary of statistical inference is provided as follows.

In an EI-Sense system, a general multi-class classifier, e.g., an $L_c$-class classifier, can be implemented using the method of one-versus-one, which decomposes the classifier into $J=L_c(L_c-1)/2$ binary classifiers~\cite{bishop2006pattern}. 
For simplicity and clarity, we perform analysis based on binary classification, which is also known as \emph{hypothesis testing}.
For a pair of classes, e.g., class-1 and class-2, the decision of feature vector $\bx$ is based on the comparison of their likelihoods $p_1(\bx)$ and $p_2(\bx)$.
We introduce the \emph{minus-log likelihood ratio} as the \emph{discriminant score function} to facilitate classification:
\begin{equation}
    \xi(\bx)\triangleq-\ln\frac{p_1(\bx)}{p_2(\bx)}=-\ln p_1(\bx)+\ln p_2(\bx).
\end{equation}
The decision rule becomes: If $\xi(\bx)<0$, then $\bx$ is classified to class-1, and vice versa.
In Gaussian mixture model, the likelihood function for samples drawn from class-$\ell$ is
\begin{equation}
    p_{\ell}(\bx)\!=\!\frac{1}{(2\pi)^{\frac{d}{2}}|\bSigma|^{\frac{1}{2}}}\exp\!\l(\!-\frac{1}{2}(\bx\!-\!\bmu_{\ell})^{\sT}\bSigma^{-1}(\bx\!-\!\bmu_{\ell})\!\r)\!,
\end{equation}
thus the discriminant score function is given as
\begin{equation}
    \xi(\bx)=(\bmu_2-\bmu_1)^{\sT}\bSigma^{-1}\bx+\frac{1}{2}(\bmu_1^{\sT}\bSigma^{-1}\bmu_1-\bmu_2^{\sT}\bSigma^{-1}\bmu_2).
\end{equation}
For illustration, the correlated two-dimensional Gaussian mixture model is shown in Fig.~\ref{Fig: GMM}, where the decision boundary is determined by $\xi(\bx)=0$. 
Since $\xi(\bx)$ is a linear transformation from a $d$-dimensional space to one-dimension, it is a Gaussian random variable when $\bx$ is a Gaussian distributed random vector given the class-$\ell$. It can be verified that 
$(\xi(\bx)|\ell=1)\sim\cN(-D,2D)$ and
    $(\xi(\bx)|\ell=2)\sim\cN(D,2D).$
The distribution of corresponding discriminant score function is presented in Fig.~\ref{Fig: Disciminant function}, where the Bayes error refers to the probability of samples from the other side of the boundary.
In this context, the sensing error is  $P_e=Q\l(\sqrt{{D}/{2}}\r)$.
In the context of sequential sensing, we suppose there are $M$ feature vectors, denoted as $\{\bx_{1},\cdots,\bx_{M}\}$, with $M\leq K$, are successfully received by the edge server.
The score function can be denoted as
\begin{align}
    \Xi(\{\bx_{k}\}_{k=1}^{M})&=-\ln\frac{p_1(\bx_{1},\cdots,\bx_{M})}{p_2(\bx_{1},\cdots,\bx_{M})}\nn\\
    &=\sum_{k=1}^{M}-\ln\frac{p_1(\bx_{k})}{p_2(\bx_{k})}=\sum_{k=1}^{M}\xi(\bx_{k}).
\end{align}
The mean and variance of the discriminant score function are $\mE[\Xi|\ell=1]=-MD$, $\mE[\Xi|\ell=2]=MD$ and $\var[\Xi|\ell=1]=\var[\Xi|\ell=2]=2MD$.
Hence, using multiple observations, the score functions becomes more separable, enhancing the inference performance.
The equivalent discriminant gain is $M$ times of the original one.

\vspace{-3mm}
\section{A Source-Channel Tradeoff}\label{Sec: Analysis}
In this section, we derive the inherent source-channel tradeoff of the ultra-low-latency EI-Sense system. Specifically, we analyze how source distortion, caused by block quantization, affects feature quality by measuring the reduction in discriminant gain. Additionally, we investigate the impact of channel reliability, degraded by packet loss, on the quantity of received features, and analyze the system performance by quantifying the resulting sensing errors. Finally, combining these results yields the desired source-channel tradeoff.

\vspace{-3mm}
\subsection{Effect of Source Distortion on Sensing}\label{Sec: Effect of Source Coding}
\subsubsection{Isotropic Gaussian Noise Approximation}
The source distortion of individual feature vectors originates from the quantization noise. 
During the quantization process, the errors induced by the uniform scalar quantizer can be approximated by adding i.i.d. uniform noise into the feature vector $\bx$, as shown below (see Fig.~\ref{Fig: Quantization noise}): 
\begin{equation}
    \tbx'=\cQ_{\Delta,U}(\tbx)=\tbx+\bn_q,
\end{equation}
where each element in $\bn_q$, i.e., $n_{q,i}=q(\tx_i)-\tx_i$, follows the i.i.d. uniform distribution: $n_{q,i}\sim\cU\l(-\frac{\Delta}{2},\frac{\Delta}{2}\r)$. 
Accordingly, the mean and variance of the quantization noise $\bn_q$ are $\mE[\bn_q]=\b0$ and $\var[\bn_q]=\frac{\Delta^2}{12}\bI$, respectively.
Hence, the source encoded feature vector for uploading is expressed as $\cF(\bx)=\bA\bx+\bn_q$.
At the server's end, the received feature vector after decoding operation is
\begin{equation}
    \hbx=\cF^{-1}(\cF(\bx))=\bA^{\sT}(\bA\bx+\bn_q)=\bx+\bA^{\sT}\bn_q,
\end{equation}
where we denote the final source noise as $\bz_q\triangleq\hbx-\bx=\bA^{\sT}\bn_q$.
\begin{lemma}[Block quantization noise]\label{Lemma: source distortion}
    The source distortion caused by block quantization, which consists of transform coding and uniform scalar quantization, can be approximated as isotropic Gaussian noise: $\bz_q\sim\cN(\b0,\sigma_q^2\bI)$. The variance of $\bz_q$ is given by
    \begin{equation}
        \sigma_q^2=\frac{\Delta^2}{12}=\frac{U^2}{3(2^R-1)^2}.\label{Eqn: variance of source distortion}
    \end{equation}
\end{lemma}
\begin{proof}
    We prove the normality and isotropy of source distortion as follows.
    \begin{itemize}
        \item[i)] \textbf{Gaussian approximation.} The $i$-th element of the source noise can be expressed as $z_{q,i}=(\bA^{\sT}\bn_q)_i=\sum_{j=1}^da_{i,j}n_{q,j}$, where $a_{i,j}$ is the $(i,j)$-th coefficient in orthogonal matrix $\bA$ satisfying $\sum_{j=1}^da_{i,j}^2=1$, and $n_{q,j}$ is the i.i.d. random variable that follows the uniform distribution $\cU(-\Delta/2,\Delta/2)$. According to the Fisher's central limit theorem (see Theorem 3.1 in~\cite{fisher1992skorohod}), the noise $z_{q,i}$ converges Gaussian distribution $\cN(0,\sigma_q^2)$ as the feature vector's dimension $d$ grows, with the variance being $\sigma_q^2=\frac{\Delta^2}{12}$.
        \item[ii)] \textbf{Isotropic noise.} The mean and variance of the quantization noise $\bn_q$ are $\mE[\bn_q]=\b0$ and $\var[\bn_q]=\sigma_q^2\bI$. Therefore, the mean of source noise $\bz_q=\bA^{\sT}\bn_q$ is $\mE[\bz_q]=\bA^{\sT}\mE[\bn_q]=\b0$, and the covariance matrix is
        \begin{align}
            \cov(\bz_q,\bz_q)&=\mE[\bz_q\bz_q^{\sT}]-\mE[\bz_q]\mE[\bz_q]^{\sT}=\mE[\bA^{\sT}\bn_q\bn_q^{\sT}\bA]\nn\\
            &=\bA^{\sT}\mE[\bn_q\bn_q^{\sT}]\bA=\sigma_q^2\bA^{\sT}\bA=\sigma_q^2\bI.
        \end{align}
    \end{itemize}
    This completes the proof.
\end{proof}

\begin{figure}[t!]
    \centering
    \subfigure[Quantization noise]{\includegraphics[width=0.48\columnwidth]{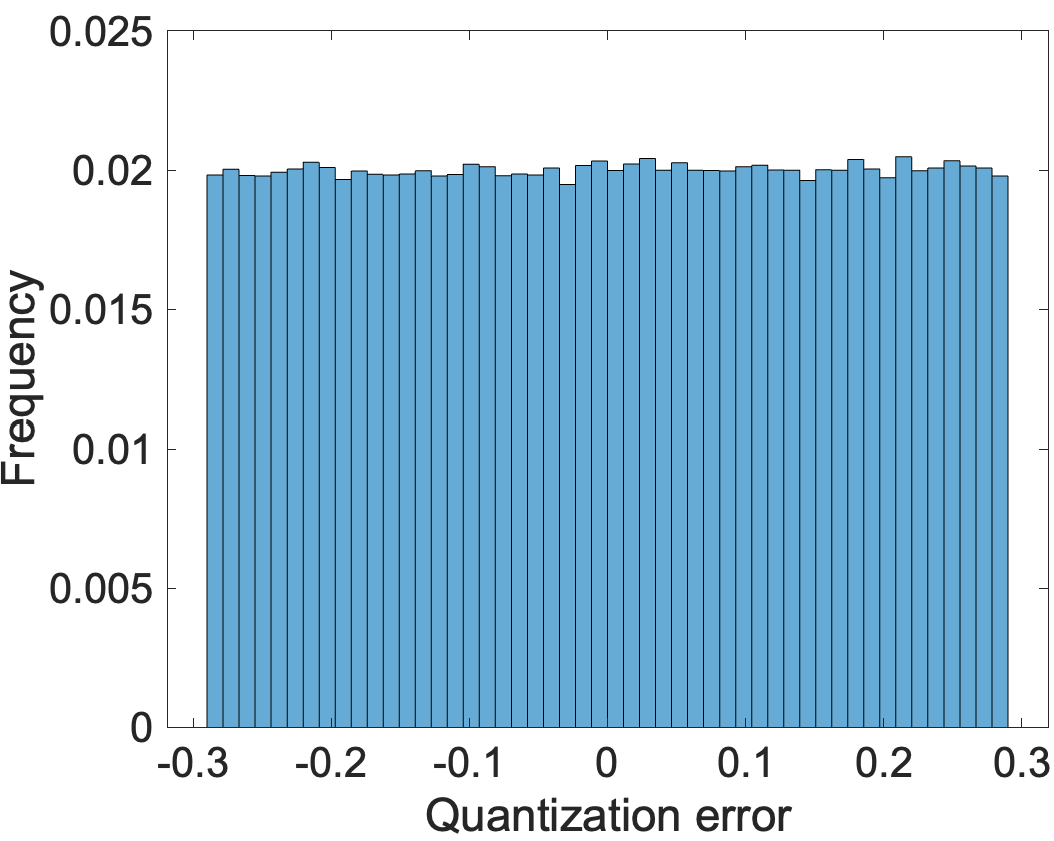}
    \label{Fig: Quantization noise}}
    \subfigure[Source distortion]{\includegraphics[width=0.48\columnwidth]{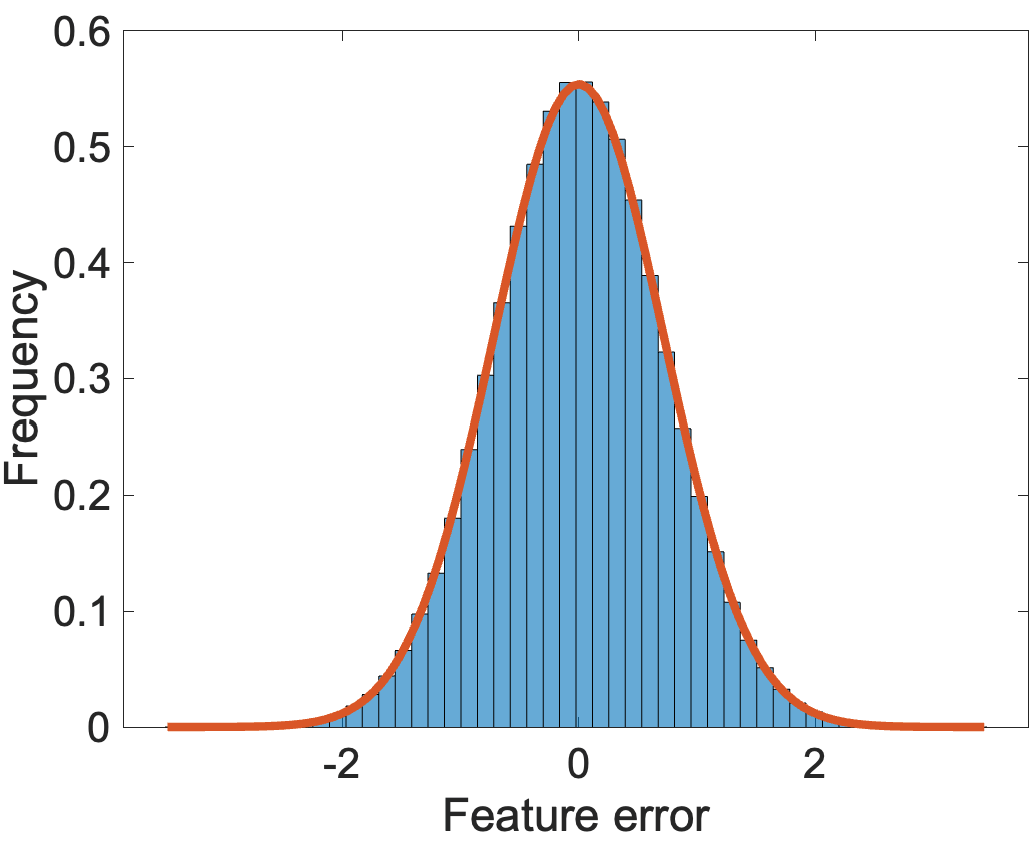}
    \label{Fig: Source distortion}}
    \caption{Histogram of quantization noise and source distortion for a single feature. The quantization level is set to 4 bits. (a) Quantization error, defined as the difference between the input and output of the uniform scalar quantizer. (b) Feature error, defined as the difference between the decoded feature and the original feature. The data is fitted using a zero-mean Gaussian distribution.\vspace{-4mm}}
\end{figure}
To illustrate the approximation, the histogram of the source distortion in one dimension is shown in Fig.~\ref{Fig: Source distortion}. The feature vector dimension is set to $d = 50$, and the approximation closely matches the observed distribution.

Without loss of generality, we assume that the feature vector $\bx$ is sampled from the cluster corresponding to class-$\ell$. The received noisy feature vector can be expressed as the sum of two independent multivariate Gaussian vectors: $\hbx = \bx + \bz_q$, where $\bx \sim \cN(\bmu_{\ell}, \bSigma)$ and $\bz_q \sim \cN(\b0, \sigma_q^2 \bI)$. Therefore, the received feature vector $\hbx$ follows the Gaussian distribution: $(\hbx | \ell) \sim \cN(\bmu_{\ell}, \bSigma + \sigma_q^2 \bI)$.

\subsubsection{Discriminant Gain Reduction}
For clarity, we denote the discriminant gain for the original feature vectors without distortion as $D_0$ which is specified in \eqref{Eqn: discriminant gain}. 
When source distortion $\bz_q \sim \cN(\b0, \sigma_q^2 \bI)$ is introduced, the effective discriminant gain for the received feature vectors becomes:
\begin{equation}
    D(\sigma_q^2)=\frac{1}{2}(\bmu_1-\bmu_2)^{\sT}(\bSigma+\sigma_q^2\bI)^{-1}(\bmu_1-\bmu_2).\label{Eqn: distorted discriminant gain}
\end{equation}

To quantify the change in discriminant gain caused by the source distortion, we present the following theorem.
\begin{theorem}[Discriminant gain reduction]\label{Thm: discriminant gain reduction}
    Consider the source distortion from block quantization as $\bz_q\sim\cN(\b0,\sigma_q^2\bI)$. The reduction in discriminant gain can be bounded by
    \begin{align}
        \sigma_q^2\tr\{\bSigma\!+\!\sigma_q^2\bI\}^{-1}\!\leq\! \frac{D_0\!-\!D(\sigma_q^2)}{D_0}\!\leq\!\sigma_q^2\tr\{(\bSigma\!+\!\sigma_q^2\bI)^{-1}\}.\label{Eqn: discriminant gain reduction}
    \end{align}
\end{theorem}
\begin{proof}
    (See Appendix~\ref{Proof: discriminant gain reduction}).
\end{proof}

The change in discriminant gain is quantified, allowing us to observe the relationship between discriminant gain reduction and source distortion. 
The distortion caused by source coding of the feature vector reduces the discriminant gain, as indicated by the relation $D_0 - D(\sigma_q^2) \geq 0$ derived from equation~\eqref{Eqn: discriminant gain reduction}, leading to $D(\sigma_q^2) \leq D_0$. 
As the discriminant gain decreases, the classifier's performance deteriorates accordingly. 
Additionally, the reduction in discriminant gain is an increasing function of the source distortion variance, $\sigma_q^2$, which comes from quantization effects. 
To maintain an acceptable sensing error, it is crucial to control the quantization levels above a specific threshold. 
With a sufficiently large number of quantization bits, i.e., $R \gg 1$, the distorted discriminant gain approximates $D(R) = D_0(1 - O(4^{-R}))$. 
Thus, as the number of quantization bits increases, the difference between $D(R)$ and $D_0$ diminishes exponentially.

\vspace{-3mm}
\subsection{Effect of Channel Reliability on Sensing}\label{Sec: Effect of Channel Coding}
Due to block fading channels, each transmission of a feature vector corresponds to a different packet loss probability. Specifically, the transmission of feature vector $\bx_{k}$ in the $k$-th time slot has a packet loss probability of $\varepsilon_{p,k}$. 
The number of successful transmissions out of a total of $K$ time slots, denoted by $M$, follows a Poisson binomial distribution: $M \sim {\sf PoiBin}(K; 1-\varepsilon_{p,1}, \ldots, 1-\varepsilon_{p,K})$. The \emph{probability mass function} (PMF) is given by
\begin{equation}
    \Pr(M=m)=\sum_{\cA\in \cF_m}\prod_{i\in\cA}(1-\varepsilon_{p,i})\prod_{j\in\cA^c}\varepsilon_{p,j},
\end{equation}
where $\cF_m$ is the set of all subsets of $m$ integers that can be selected from $\{1,2,\cdots,K\}$.
The score function $\Xi=\sum_{k=1}^{M}\xi(\bx_{k})$ is a random sum, where both $M$ and $\xi(\bx_k)$ are random variables. 
The sensing performance can be characterized by the following theorem.
\begin{theorem}[Sensing performance]\label{Thm: performance of sequential sensing}
    Consider each observation contributes a discriminant gain of $D$. When $K$ observations are transmitted, each with an individual packet loss probability $\varepsilon_{p,k}$, the sensing error probability can be upper bounded by
    \begin{equation}
        P_e\leq\l(\exp\l(-\frac{1}{4}D\r)+\l(1-\exp\l(-\frac{1}{4}D\r)\r)\overline{\varepsilon_p}\r)^K,\label{Eqn: error of sequential sensing}
    \end{equation}
    where $\overline{\varepsilon_p}=\frac{1}{K}\sum_{k=1}^K\varepsilon_{p,k}$ is average packet loss probability.
\end{theorem}
\begin{proof}
    (See Appendix~\ref{Proof: performance of sequential sensing}).
\end{proof}

The above result suggests that the sensing error probability can be reduced by enhancing the discriminant gain of the transmitted feature vectors, decreasing the packet loss probability during transmission, or increasing the number of diverse feature vector transmissions.

The packet loss probability is related to the receive SNR as shown in \eqref{Eqn: packet loss probability}.
For Rayleigh fading channel, the receive SNR, $\gamma_k=\gamma_0\|\bh_{k}\|^2$, is a sum of $L$ i.i.d. random variables, each having an exponential distribution $|h_k^{(l)}|^2\sim{\sf Exp}(1)$.
Therefore, $\gamma_k$ follows the Erlang distribution: $\gamma_{k}\sim{\sf Erlang}(L,1/\gamma_0)$, with PDF as shown below:\vspace{-1mm}
\begin{equation}
    p_{\gamma}(\gamma)=\frac{(\gamma/\gamma_0)^{L-1}\exp(-\gamma/\gamma_0)}{\gamma_0\Gamma(L)},\vspace{-2mm}\label{Eqn: Erlang distribution}
\end{equation}
where $\Gamma(\cdot)$ denotes the Gamma function.
The average packet loss probability can be estimated as 
\begin{align}
    \overline{\varepsilon_{p}}\doteq\int_{0}^{\infty}Q\l(\sqrt{\frac{N}{V(\gamma)}}\l(C(\gamma)-R_c\r)\r)p_{\gamma}(\gamma)\d\gamma.\label{Eqn: expected packet loss probability}
\end{align}
Further considering the coding rate $R_c$, the average packet loss probability, which is a function of the coding rate and received SNR, can be characterized in the following proposition.

\begin{figure}[t!]
    \centering
    \subfigure[Packet loss vs Transmit SNR]{\includegraphics[width=0.48\columnwidth]{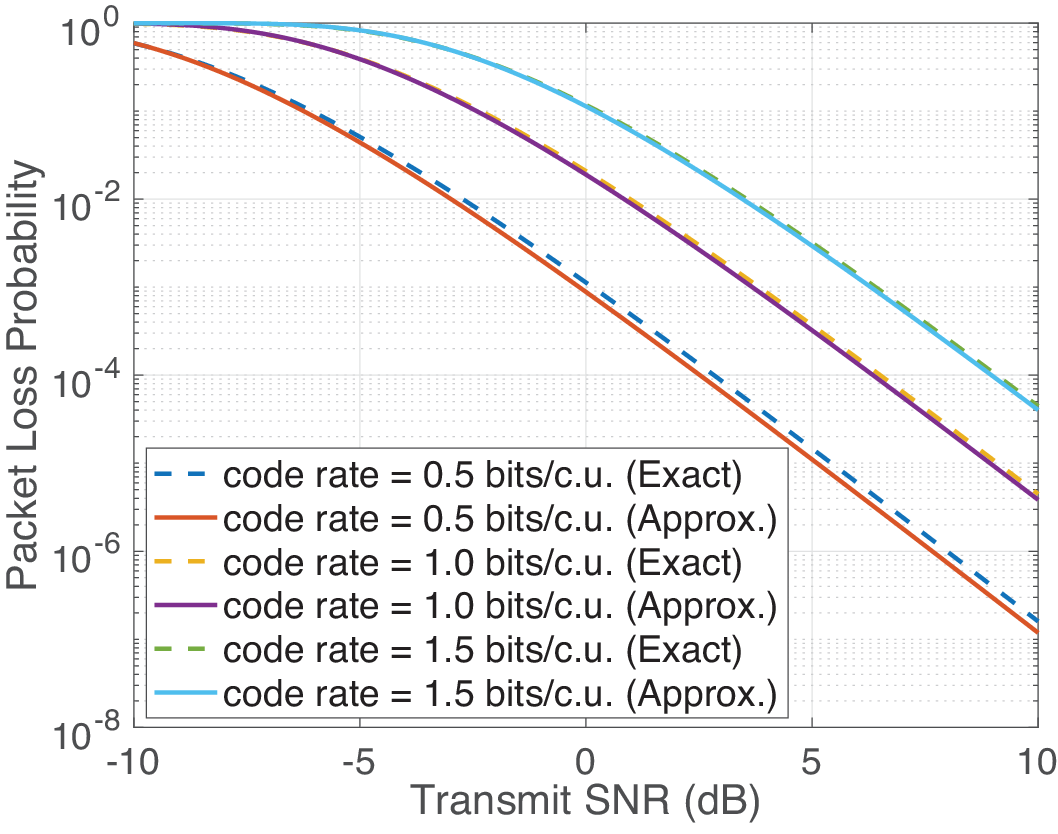}}
    \subfigure[Packet loss vs Coding rate]{\includegraphics[width=0.48\columnwidth]{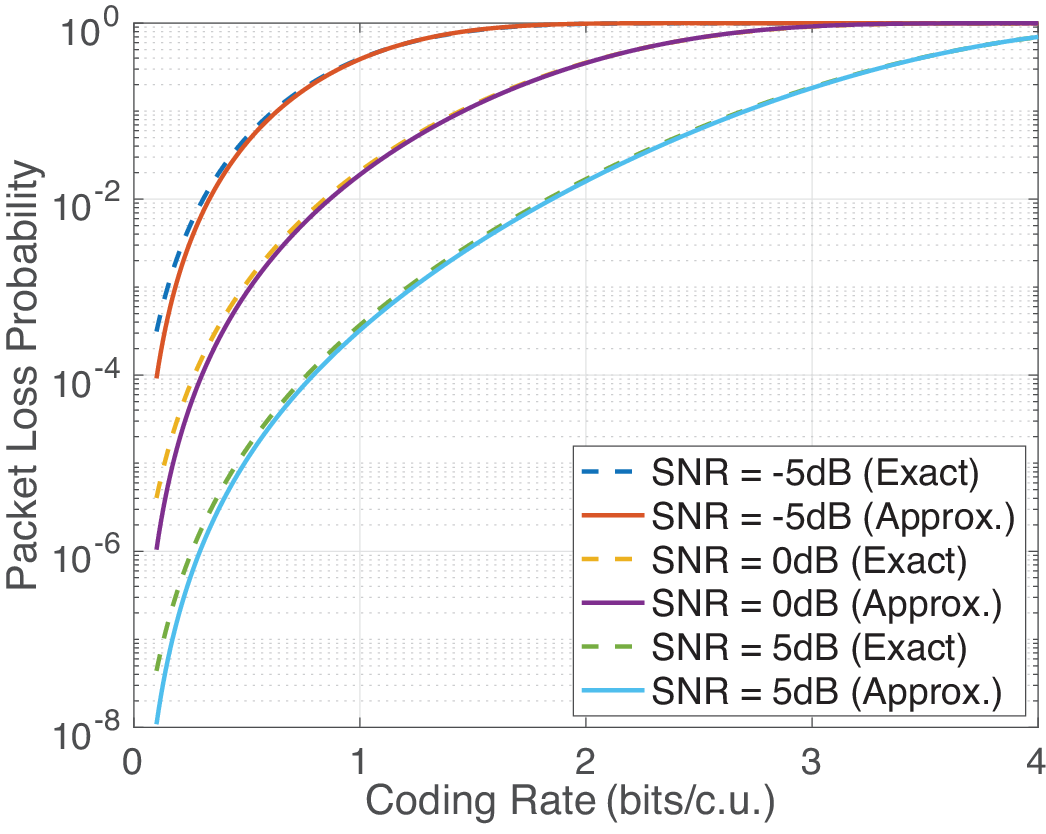}}\vspace{-2mm}
    \caption{Comparison between the exact and approximated average packet loss probability as a function of (a) transmit SNR and (b) coding rate. The parameters are set to $L=4$ and $N=100$.\vspace{-5mm}}\label{Fig: approximation of average packet loss probability}
\end{figure}

\begin{proposition}[Average packet loss probability]\label{Proposition: average packet loss}
    The average packet loss probability $\overline{\varepsilon_{p}}$ can be approximated by
    \begin{align}
        \overline{\varepsilon_p} = \!\l[1-\frac{1}{\Gamma(L)}\Gamma\!\l(L,\frac{1}{\gamma_0}(2^{R_c}\!-1)\!\r)\r]\!\l\{1+O\!\l(\frac{1}{N}\r)\r\},\label{Eqn: average packet loss probability}
    \end{align}
    where $\Gamma(\cdot,\cdot)$ denotes the upper incomplete Gamma function.
\end{proposition}
\begin{proof}
    (See Appendix~\ref{Proof: average packet loss}).
\end{proof}

The approximation is typically accurate, with a relative error of $O(N^{-1})$. 
As shown in Fig.~\ref{Fig: approximation of average packet loss probability}, the approximation in \eqref{Eqn: average packet loss probability} is compared with its ground truth, demonstrating that the approximation error is generally small.

Let $\beta\triangleq\frac{1}{\gamma_0}(2^{R_c}-1)$. One can verify that
    $\lim_{L\to\infty}\frac{\ep}{\beta^L/L!}=e^{-\beta},$
and thus the asymptotic behavior of the sensing error for large $L$ $(L\gg1)$ becomes:
\begin{equation}
    P_e\leq \exp\!\l(-\frac{1}{4}D\r)^K\l\{1+O\!\l(\frac{(2^{R_c}-1)^L}{\gamma_0^LL!}\r)\r\}.
\end{equation}
We note that, regardless of whether $\beta < 1$ or $\beta \geq 1$, the order term $O(\beta^L / L!)$ always decreases to zero asymptotically.
\begin{remark}[Channel diversity vs feature diversity]
    More antennas at the edge server allow for increased spatial diversity in this SIMO system, which enhances the reliability of transmitting feature vectors. This improved reliability results in more feature vectors being successfully decoded at the edge server for inference, thereby reducing the sensing error.
\end{remark}

\vspace{-5mm}
\subsection{Source-Channel Tradeoff}\label{Sec: Source-Channel Tradeoff}
As discussed, the performance of the EI-Sense system is influenced by both source distortion and channel reliability.
The source distortion, introduced by block quantization, affects the quality of the feature vectors, as discussed in Section~\ref{Sec: Effect of Source Coding}. 
The discriminant gain reduction caused by source distortion, quantified by Theorem~\ref{Thm: discriminant gain reduction}, indicates that a larger $\sigma_q^2$ leads to a decrease in the sensing performance. 
To mitigate this, increasing the number of quantization bits $R$, can reduce the distortion, thereby preserving the discriminant gain.
On the other hand, the reliability of the communication channel, influenced by packet loss probability $\varepsilon_p$, determines the quantity of successfully received feature vectors. 
As analyzed in Section~\ref{Sec: Effect of Channel Coding}, the sensing error probability decreases with enhanced channel reliability, which can be achieved by reducing coding rate or improving the receive SNR as characterized in Proposition~\ref{Proposition: average packet loss}.
The tradeoff between source and channel coding emerges from the need to balance feature quality and quantity. 
Reducing the channel coding rate improves the probability of successful transmission, allowing more feature vectors to be received. 
However, this comes at the expense of a higher source distortion due to fewer bits being allocated for quantization. 
Conversely, allocating more bits to source coding reduces distortion but increases the likelihood of packet loss due to a higher coding rate.
To achieve the minimum sensing error probability, it is crucial to find the optimal balance between source and channel coding. 

\begin{figure}[t!]
    \centering
    \includegraphics[width=0.58\columnwidth]{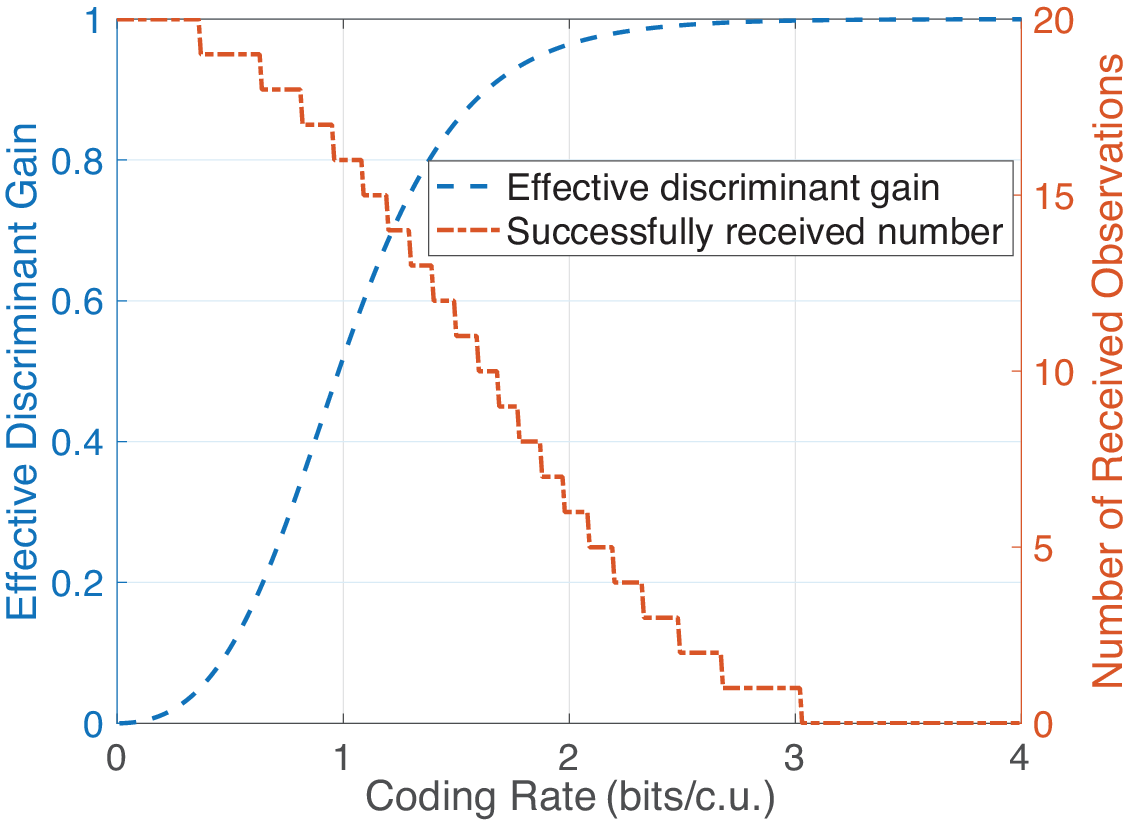}\vspace{-3mm}
    \caption{An illustration of the source-channel tradeoff. Both the effective discriminant gain and the number of successfully received observations are plotted as a function of the coding rate. The parameters are set to $L=2$, $\gamma_0=1$ dB, $N=100$, $d=50$, $D_0=1$, $U=5$, and $K=20$.\vspace{-5mm}}\label{Fig: SC tradeoff}
\end{figure}

\vspace{-3mm}
\section{Optimal Coding Rate Adaptation}\label{Sec: Tradeoff and optimization}
In this section, we firstly discuss the pivotal role of coding rate in managing the source-channel tradeoff and its impact on sensing performance. 
Then, we introduce an algorithm designed to determine the optimal coding rate, utilizing a tractable surrogate function with proven concavity.

\vspace{-3mm}
\subsection{Role of Coding Rate}\label{Sec: Effect of Source and Channel Coding}
The sensing error probability depends on two key factors: i) the effective discriminant gain, $D(\sigma_q^2)$, which is a function of the source distortion noise variance, $\sigma_q^2$, and ii) the average packet loss probability, $\overline{\varepsilon_p}(R_c)$, which depends on the coding rate, $R_c$. 
We emphasize that the variance $\sigma_q^2$ is determined by the quantization resolution, $\Delta$, as shown in \eqref{Eqn: variance of source distortion}, and it has a one-to-one relationship with the coding rate $R_c$ under a fixed latency constraint (i.e., a given number of channel uses, $N$), as derived from \eqref{Eqn: code rate} and \eqref{Eqn: variance of source distortion}.
As a result, both the effective discriminant gain and the average packet loss probability are functions of the coding rate, making the sensing error probability $P_e$ act as a function of $R_c$. This relationship is expressed as $P_e(R_c)$.
\begin{lemma}[Monotonicity of discriminant gain and packet loss probability]
As functions of the coding rate $R_c$, both the effective discriminant gain $D(R_c)$ and the average packet loss probability $\ep(R_c)$ are strictly increasing.
\end{lemma}
The lemma can be proven by verifying that the derivatives $D'(R_c) > 0$ and $\ep'(R_c) > 0$, which is straightforward and hence omitted for brevity.
This lemma highlights the inherent tradeoff between source distortion and channel reliability as mediated by the coding rate $R_c$. 
A higher coding rate implies more quantization bits per feature vector. On the one hand, this leads to a more accurate representation of the feature vector, enhancing the effective discriminant gain and reducing sensing error. On the other hand, a higher coding rate increases the packet loss probability, which reduces the number of successfully received feature vectors and, in turn, increases the sensing error.
As a result, the sensing error probability exhibits a fundamental tradeoff: increasing the coding rate $R_c$ improves the accuracy of feature representation but simultaneously degrades channel reliability. 
The tradeoff is illustrated in Fig.~\ref{Fig: SC tradeoff}.
Consequently, there exists an optimal coding rate, $R_c^{\star}$, that minimizes the sensing error by balancing these opposing effects.
The problem is formulated as follows:
\begin{align}
    \min_{R_c}\quad&P_e(R_c)\label{Eqn: objective}\\
    \text{s.t.}\quad&R={R_cN}/{d}\in\mN_+,\tag{\ref{Eqn: objective}.C1}\\
    &N\in\{1,\cdots,N_{\max}\}.\tag{\ref{Eqn: objective}.C2}
\end{align}
The constraint (\ref{Eqn: objective}.C1) ensures that the number of quantization bits per feature is a positive integer, while the constraint (\ref{Eqn: objective}.C2) represents the radio resource limitation, defined by the restricted number of available channel uses.

\vspace{-3mm}
\subsection{Coding Rate Optimization}
Problem \eqref{Eqn: objective} is a non-convex problem, which makes finding the optimum challenging. 
To address this, we relax the integer constraint (\ref{Eqn: objective}.C1), allowing the coding rate $R_c$ to take positive real values.
Then, a sufficient condition for the optimal coding rate in this case is given as follows:
\begin{align}
    R_c^{\star}\!=\!\arg\!\min_{R_c>0}\exp\!\l(\!-\frac{D(R_c)}{4}\!\r)\!+\!\l(\!1\!-\!\exp\!\l(\!-\frac{D(R_c)}{4}\!\r)\!\r)\!\overline{\varepsilon_p}(R_c).\label{Eqn: optimal code rate}
\end{align}
However, this problem remains non-convex, and deriving a closed-form expression for the optimal $R_c^{\star}$ is still infeasible.
\subsubsection{Surrogate function}
To tackle the challenge, we introduce a surrogate function to reformulate the original problem. 
Specifically, problem \eqref{Eqn: optimal code rate} is re-expressed as the following equivalent problem:
\begin{equation}
    R_c^{\star}=\arg\max_{R_c>0}\phi(R_c),\label{Eqn: equivalent problem}
\end{equation}
where the surrogate function $\phi(\cdot)$ is defined as
\begin{equation}
    \phi(R_c)=\ln\l[\l(1-\exp\l(-\frac{D(R_c)}{4}\r)\r)(1-\overline{\varepsilon_{p}}(R_c))\r].\label{Eqn: surrogate function}
\end{equation}

To analyze the convexity or concavity of the surrogate function in \eqref{Eqn: surrogate function}, we investigate the properties of the effective discriminant gain function $D(R_c)$ and the average packet loss function $\ep(R_c)$ as follows:

i) For effective discriminant gain $D(R_c)$, we detail the concavity of this function as described below.
\begin{lemma}[Concavity of effective discriminant gain]\label{Lemma: convexity of DG}
    The effective discriminant gain $D(R_c)$, as defined in \eqref{Eqn: distorted discriminant gain}, is a concave function of the coding rate $R_c$.
\end{lemma}
\begin{proof}
    (See Appendix~\ref{Proof: convexity of DG}).
\end{proof}

ii) For average packet loss probability $\ep(R_c)$, it can be shown that this function is generally non-convex. 
However, we observe a useful property of the average packet success probability, defined as $(1-\ep(R_c))$, which is described below.
\begin{lemma}[Log-concavity of average packet success probability]\label{Lemma: convexity of epsilon}
    The average packet success probability $(1-\overline{\varepsilon_p}(R_c))$, with $\overline{\varepsilon_p}(R_c)$ given in \eqref{Eqn: average packet loss probability}, is log-concave with respect to the coding rate $R_c$.
\end{lemma}
\begin{proof}
    (See Appendix~\ref{Proof: convexity of epsilon}).
\end{proof}

To facilitate our analysis, we define the two functions:
\begin{align}
    \phi_1(R_c)&=\ln\l(1-\exp\l(-\frac{1}{4}D(R_c)\r)\r),\\
    \phi_2(R_c)&=\ln\l(1-\overline{\varepsilon_p}(R_c)\r),
\end{align}
such that the surrogate function in \eqref{Eqn: equivalent problem} can be denoted as $\phi(R_c)=\phi_1(R_c)+\phi_2(R_c)$.
Based on the Lemmas \ref{Lemma: convexity of DG} and \ref{Lemma: convexity of epsilon}, the concavity of the surrogate function is established in the following proposition.
\begin{proposition}[Concavity of surrogate function]\label{Proposition: convexity}
The surrogate function $\phi(R_c)$ defined in \eqref{Eqn: surrogate function} is concave.
\end{proposition}
\begin{proof}
    Considering $\phi(R_c)=\phi_1(R_c)+\phi_2(R_c)$, we prove the convexity of functions $\phi_1(R_c)$ and $\phi_2(R_c)$ as follows:
\begin{itemize}
    \item[i)] It is easy to verify that the function $f(x)=\ln(1-x)$, defined at $x\in(0,1)$, is concave and non-increasing. From Lemma~\ref{Lemma: convexity of DG}, we have $D(R_c)$ is concave such that the function $g(R_c)=\exp\l(-\frac{1}{4}D(R_c)\r)$ is convex. Therefore, the composition function $\phi_1(R_c)=f(g(R_c))$ is concave.
    \item[ii)] From Lemma~\ref{Lemma: convexity of epsilon}, we have $(1-\overline{\varepsilon_p}(R_c))$ is log-concave such that  $\phi_2(R_c)=\log(1-\overline{\varepsilon_p}(R_c))$ is a concave function.
\end{itemize}
This completes the proof.
\end{proof}
\begin{figure}[t!]
    \centering
    \includegraphics[width=0.58\columnwidth]{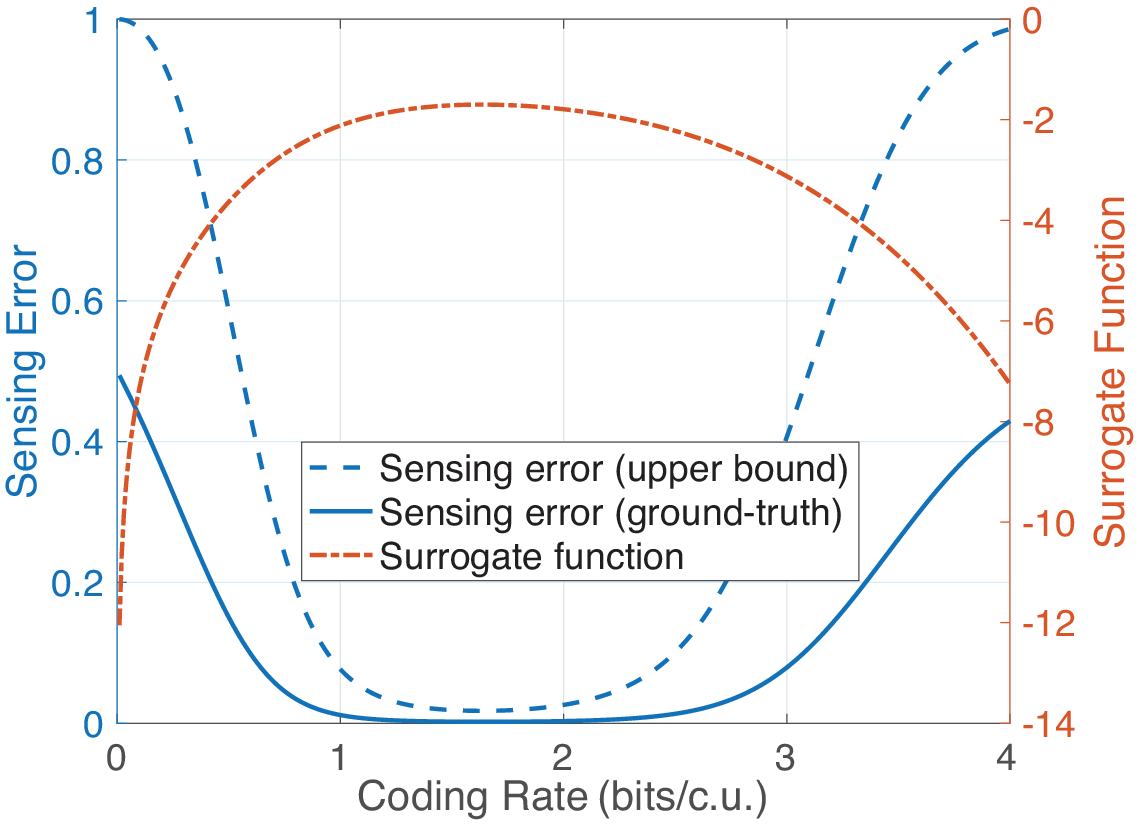}\vspace{-3mm}
    \caption{Comparison of the upper bound of sensing error, ground-truth sensing error, and the proposed surrogate function as functions of coding rate. The parameters are set to $L=4$, $\gamma_0=1$ dB, $N=100$, $d=50$, $D_0=1$, $U=5$, and $K=20$.\vspace{-5mm}}\label{Fig: Effect of code rate}
\end{figure}

Fig.~\ref{Fig: Effect of code rate} illustrates the comparison between the ground-truth sensing error, the upper bound of the sensing error, and the surrogate function across different coding rates. 
The plot demonstrates that the optimal coding rate is achieved at the same point for all three metrics, thereby validating the accuracy of the surrogate function. 
Furthermore, the concave nature of the surrogate function is clearly depicted.

\subsubsection{Approximate gradient ascent}
The derivative of $\phi_1(R_c)$ can be obtained in closed form as
\begin{equation}
    \phi_1'(R_c)=\frac{D'(R_c)}{4\l(\exp\l(\frac{1}{4}D(R_c)\r)-1\r)},\label{Eqn: Expression of phi1 prime}
\end{equation}
where $D(R_c)$ is given in \eqref{Eqn: discriminant gain}, and the expression of $D'(R_c)$ is provided as follows:
\begin{equation}
    D'(R_c)\!=\!\frac{NU^22^{R(R_c)}\!\ln2}{3(2^{R(R_c)}\!-\!1)^3d}\|(\bSigma+\sigma_q^2(R_c)\bI)^{-1}\!(\bmu_1\!-\!\bmu_2)\|_2^2,\!\!\label{Eqn: derivative of DG}
\end{equation}
with $R(R_c)=\frac{NR_c}{d}$ and $\sigma_q^2(R_c)=\frac{U^2}{3(2^{NR_c/d}-1)^2}$. 

The derivative of $\phi_2(R_c)$ involves $\overline{\varepsilon_p}(R_c)$ and $\overline{\varepsilon_p}'(R_c)$, which appears as an integral, making the numerical calculation difficult. 
To address this, we use the approximation as follows:
\begin{equation}
    \phi_2'(R_c)\doteq-\frac{2^{R_c}(2^{R_c}-1)^{L-1}\ln2}{\gamma_0^L\Gamma\l(L,\frac{1}{\gamma_0}(2^{R_c}-1)\r)}\exp\l(-\frac{1}{\gamma_0}(2^{R_c}-1)\r).\vspace{-2mm}\label{Eqn: Approximation of phi2 prime}
\end{equation}
The relative error of this approximation is $O(1/N)$.

Thus, the gradient of the surrogate function is estimated as 
\begin{equation}
    \phi'(R_c)=\phi_1'(R_c)+\phi_2'(R_c),
\end{equation}
where $\phi_1'(R_c)$ and $\phi_2'(R_c)$ are presented in closed-form by \eqref{Eqn: Expression of phi1 prime} and \eqref{Eqn: Approximation of phi2 prime}, respectively.

Using the gradient ascent method, the updating rule is $R_c\gets R_c+\eta\phi'(R_c)$, where $\eta$ is the step size and $\phi'(R_c)$ is the estimated gradient in current step. Finally, we can numerically determine the optimal coding rate $R_c^{\star}$. 
However, due to the approximate nature of the gradient, an error is introduced, resulting in a deviation from the true optimal solution. 
This deviation is on the order of $O(1/N)$. 
Given typical block lengths such as $N=100$, this deviation can be maintained within acceptable tolerance levels.
As illustrated in Fig.~\ref{Fig: Convergence analysis},  upon convergence, the deviation of the coding rate derived from the approximate gradient ascent, compared to the coding rate derived using the ground-truth gradient, is only 0.26\%.

Further considering the integer constraint on $R$ from (\ref{Eqn: objective}.C1) and the maximum channel uses from (\ref{Eqn: objective}.C2), the number of quantization bits per feature is determined as $R^{\star}=\max\{\lfloor R_cN_{\max}/d\rceil,1\}$, where $\lfloor \cdot \rceil$ denotes rounding $R^{\star}$ to the nearest integer. 
This ensures that $R^{\star}$ satisfies the practical requirement of being a positive integer.
Finally, the coding rate is adapted to be $R_c^{\star}=R^{\star}d/N_{\max}$, which effectively balances tradeoff between source distortion and channel reliability.
\begin{figure}[t!]
    \centering
    \includegraphics[width=0.52\columnwidth]{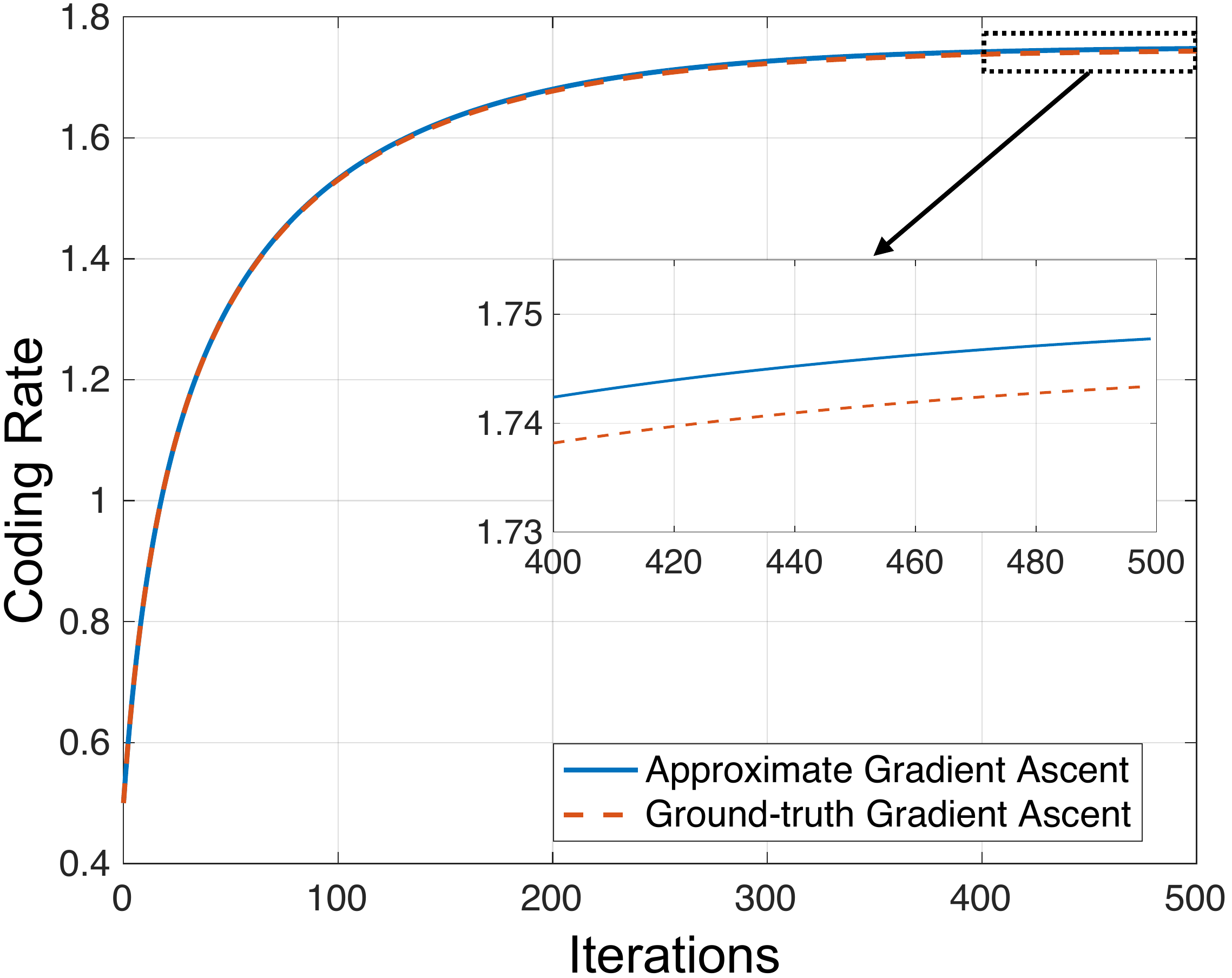}\vspace{-2mm}
    \caption{Convergence performance of the proposed approximate gradient ascent method. The parameters are set to $\eta=0.01$, $L=4$, $\gamma_0=2$ dB, $N=100$, $d=50$, $D_0=1$, $U=5$, and $K=20$.\vspace{-4mm}}\label{Fig: Convergence analysis}
\end{figure}
\vspace{-3mm}
\section{Experimental Results}\label{Sec: Experimental Results}

\subsection{Experimental Setup}
Unless specified otherwise, the default experimental settings are specified as follows.
\subsubsection{System and communication settings}
We consider a frequency non-selective Rayleigh fading channel, where the channel vector consists of i.i.d. complex Gaussian elements following $\cC\cN(0,1)$. The coherence duration, defined as the time over which the channel remains constant, spans $N=100$ symbols, which is set to transmit one feature vector per time slot. The transmit SNR is set as $\gamma_0=2$ dB.
In the system, a single-antenna sensor sequentially observes the object for $K=10$ times, and transmits the feature vectors to the edge server. The edge server is equipped with $L=4$ antennas for receiving over 10 consecutive time slots, but it successfully decodes $M\leq K$ feature vectors due to packet loss.


\subsubsection{Sensing and classification settings}
We consider both the cases of statistical inference on synthetic data and CNN-based classification on real-world data as follows.
\begin{itemize}
    \item \emph{Statistical inference on synthetic data:} In this setting, local feature vectors are generated from a Gaussian mixture model and transmitted to the classifier via short-packet communication. The feature vectors have a dimensionality of $d = 50$.
    The centroid of one cluster is a vector with all elements equal to $+0.1$, while the centroid of the other cluster is a vector with all elements equal to $-0.1$. 
    The covariance matrix is defined as $\bSigma=\bI_d$.
    The results for each sensing error are generated from 10,000 Monte Carlo experiments.
    \item \emph{Non-linear CNN-based classification on real-world data:} This setting uses the ModelNet-40 dataset~\cite{su2015multi}, which contains 40-class multi-view images of 3D objects, and the VGG-11 convolutional neural network~\cite{simonyan2015very}. The VGG-11 model is split into two components: the feature extractor which runs on the sensor, and the classifier which runs on the edge server. To reduce communication overhead, each ModelNet image is resized from its original size of $3 \times 224 \times 224$ pixels to $3 \times 56 \times 56$ pixels before being processed by the on-sensor feature extractor, which outputs a $512\times1\times1$ tensor. This tensor is then further compressed by a fully-connected layer to obtain an output feature vector of dimension $d=50$.
\end{itemize}
\subsubsection{Benchmarks}
We evaluate the proposed method against the following four benchmarks:
\begin{itemize}
    \item \emph{Brute-force search:} The optimal coding rate is determined through exhaustive search over all feasible quantization levels. This approach guarantees the minimum sensing error by identifying the globally optimal solution.
    \item \emph{URLLC:} The coding rate is selected as the highest rate that satisfies a decoding error probability threshold of $10^{-5}$, which is a standard requirement in URLLC~\cite{durisi2016toward}. Accordingly, the number of quantization bits per feature is determined by the coding rate.
    \item \emph{Full resolution:} A resolution of 32 bits is used as the baseline for full-resolution features. This choice reflects the common practice in CNNs, where feature maps are typically represented as tensors of floating-point numbers, with the float-32 data type being the most widely used.
    \item \emph{Half-bits resolution:} A resolution of 16 bits, i.e., a half of full-resolution bits, is used for representing each feature.
\end{itemize}

\begin{figure}[t!]
    \centering
    \subfigure[Effect of observation number]{\includegraphics[width=0.48\columnwidth]{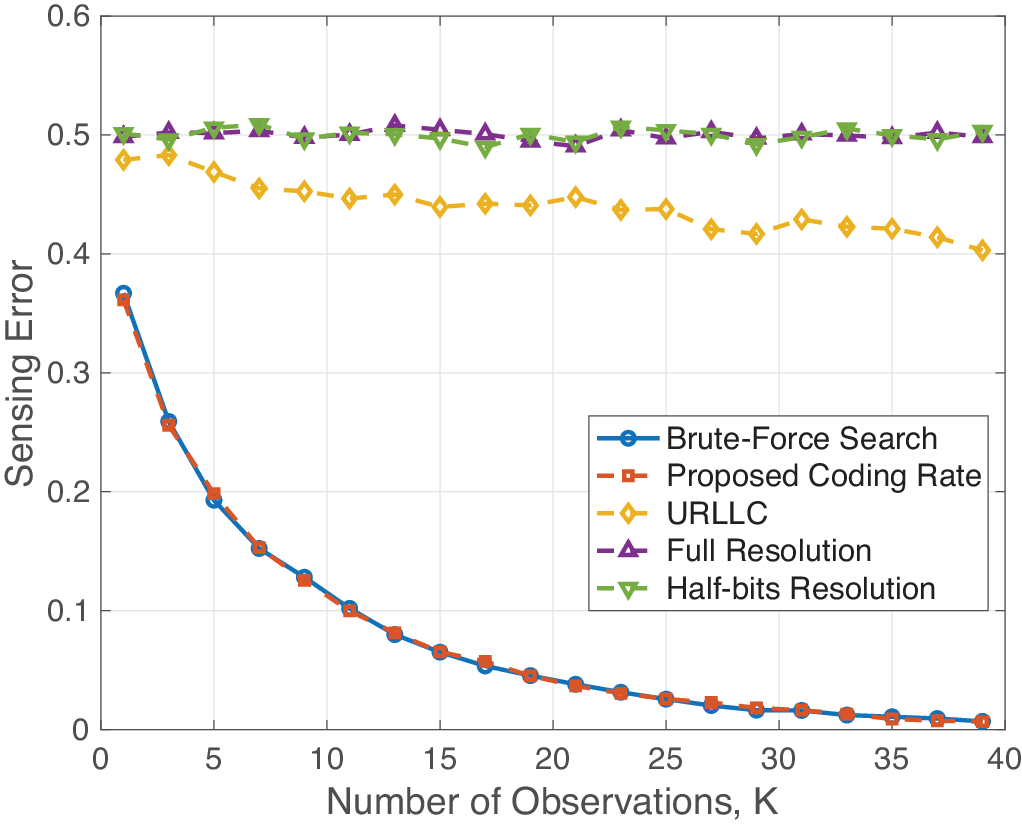}
    \label{Simulation: Effect of observation number}}
    \subfigure[Effect of transmit SNR]{\includegraphics[width=0.48\columnwidth]{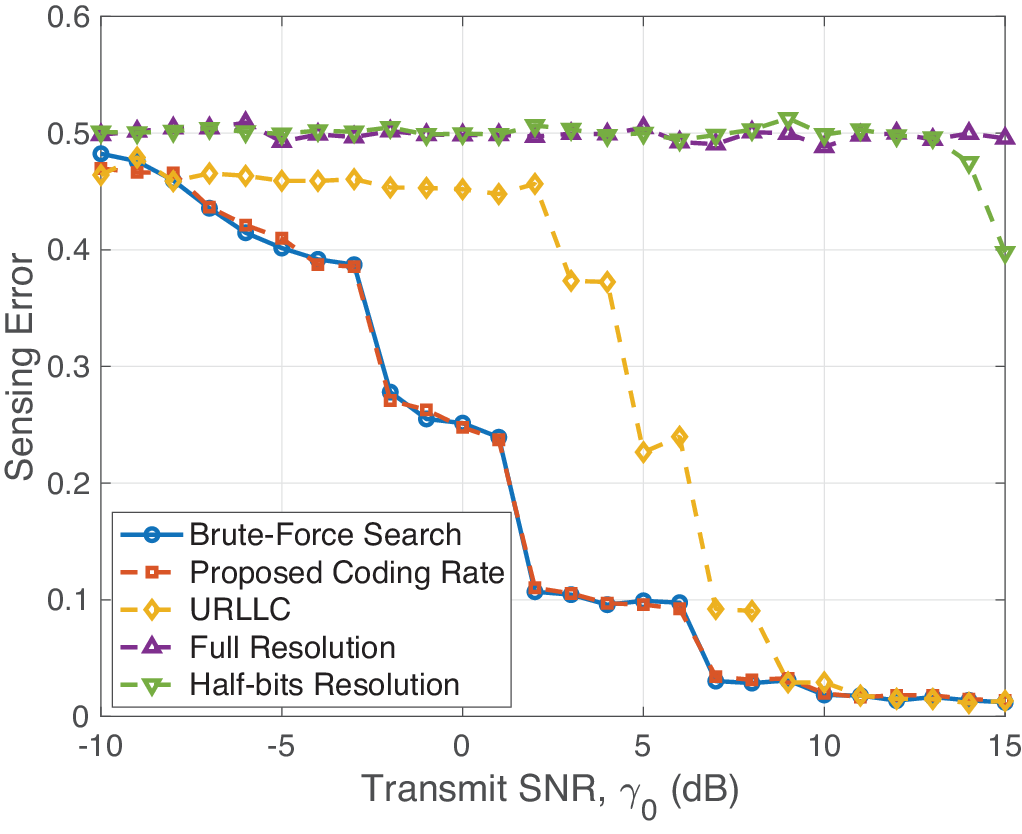}
    \label{Simulation: Effect of transmit SNR}}
    \subfigure[Effect of receive antennas]{\includegraphics[width=0.48\columnwidth]{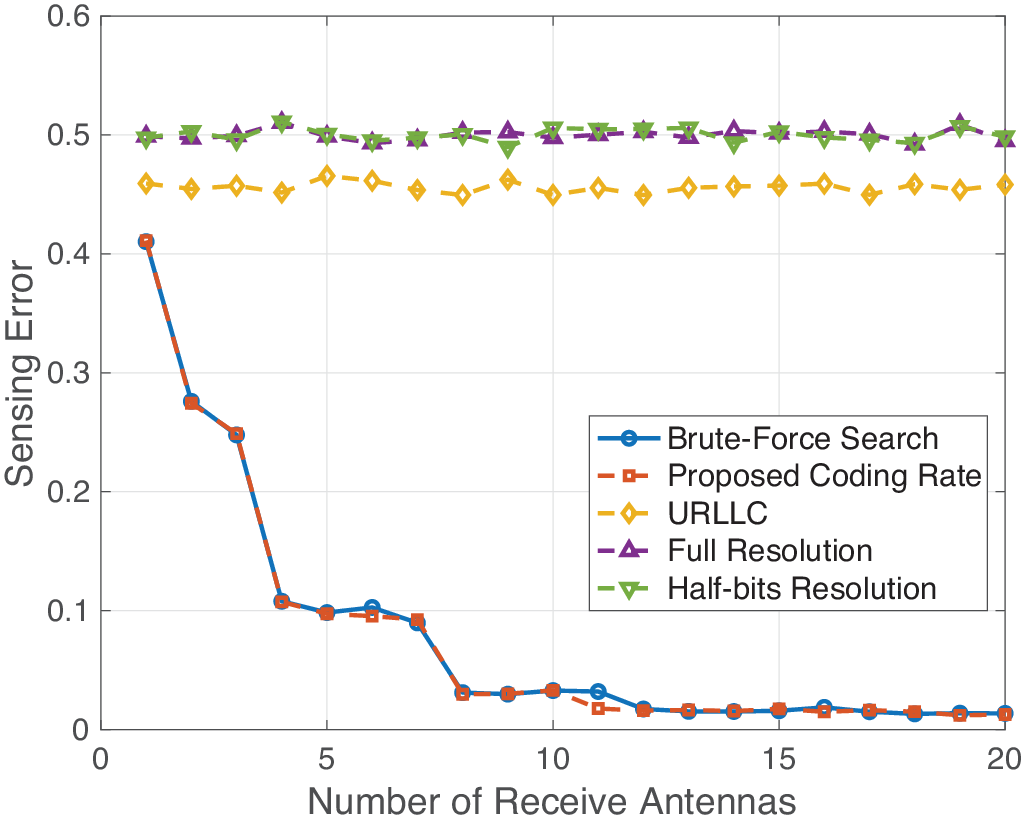}
    \label{Simulation: Effect of receive antennas}}
    \subfigure[Effect of blocklength]{\includegraphics[width=0.48\columnwidth]{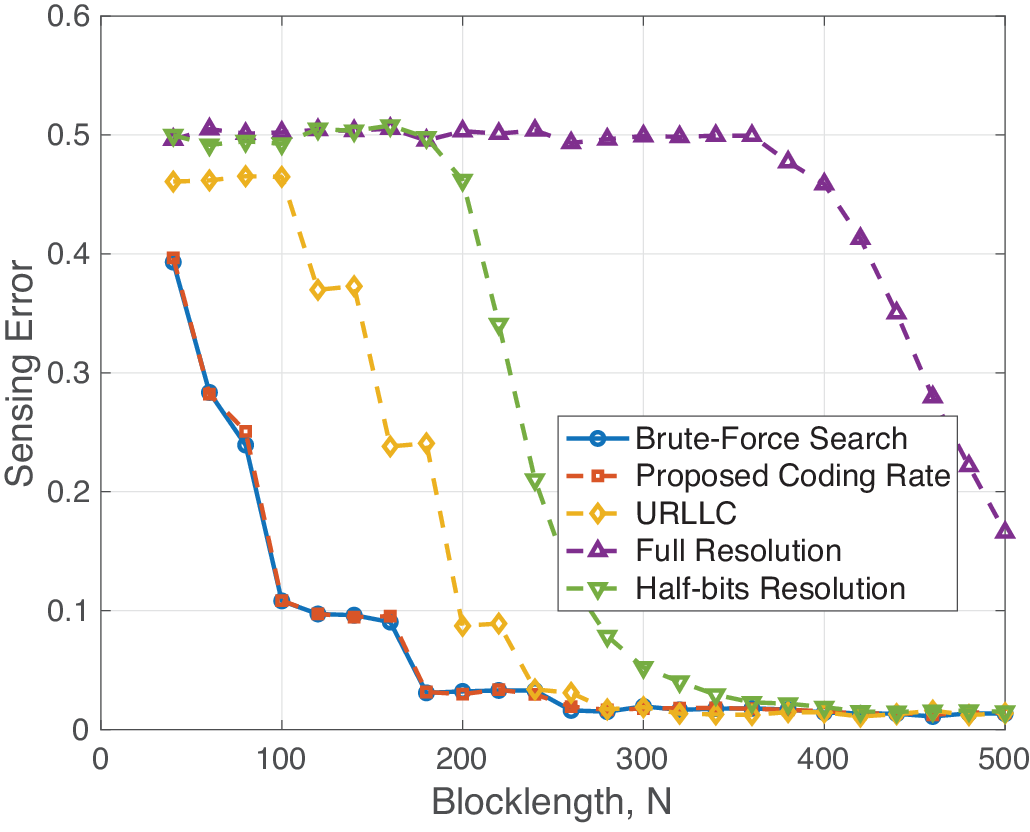}
    \label{Simulation: Effect of blocklength}}
    \caption{The effects of different parameters on sensing error and comparison with benchmarks in the case of statistical inference.\vspace{-5mm}}\label{Simulation: Linear classification}
\end{figure}

\vspace{-4mm}
\subsection{Statistical Inference Case}
The impacts of various parameters on E2E sensing error are analyzed in Fig.~\ref{Simulation: Linear classification}. 
First, as shown in Fig.~\ref{Simulation: Effect of observation number}, the sensing error probability for a single observation ($K = 1$) is close to random guessing due to the overlapping of sample clusters, underscoring the importance of multiple observations. As the number of observations, $K$, increases, the sensing error diminishes exponentially, which aligns with the theoretical results in \eqref{Eqn: error of sequential sensing}. However, this reduction comes at the cost of a linear increase in latency, as more time slots are required for feature vector transmission. 
Similarly, the sensing error decreases with higher transmit SNR, $\gamma_0$, as shown in Fig.~\ref{Simulation: Effect of transmit SNR}, which is consistent with the analytical findings. This improvement, however, comes with increased power consumption at the sensor. 
The impact of the number of receive antennas, $L$, as shown in Fig.~\ref{Simulation: Effect of receive antennas}, is also evident: as $L$ grows, channel diversity improves, reducing the packet loss probability and increasing the number of successfully decoded feature vectors, thereby lowering the sensing error. 
Finally, the blocklength, $N$, significantly affects performance as shown in Fig.~\ref{Simulation: Effect of blocklength}. Larger $N$ allows for a higher quantization bits given coding rate, reducing source distortion, and also decreases the packet loss probability. From our analysis, the dominant factor is the reduction in source distortion. However, the gains from a larger blocklength come at the cost of increased latency, which scales proportionally with the number of channel uses.

Our proposed adaptive coding rate design demonstrates effectiveness by achieving performance close to brute-force search while outperforming other benchmarks. 
Traditional reliability-centric URLLC techniques employ low coding rates to ensure high transmission reliability, but this approach severely limits the number of quantization bits per feature, causing significant source distortion and thus high sensing error probabilities. 
Conversely, high-resolution features require a high coding rate for transmission, which dramatically increases the packet loss probability. 
In such scenarios, sensing performance degrades to random guessing due to the absence of successfully received feature vectors. 
By balancing source distortion and channel reliability, our adaptive coding rate design avoids these extremes, thereby achieving superior E2E sensing performance.
Notably, the performance of other benchmark approaches begins to converge with that of our proposed scheme when the transmit SNR is high and the blocklength is long. 
This is because high transmit SNR and extended blocklength allow for the transmission of more bits (i.e., a feature vector with high resolution) with a low packet loss probability (i.e., ultra-reliable transmission).
Therefore, we conclude that our scheme is particularly advantageous in scenarios characterized by low SNR and short packet lengths.
\begin{figure}[t!]
    \centering
    \subfigure[Effect of observation number]{\includegraphics[width=0.48\columnwidth]{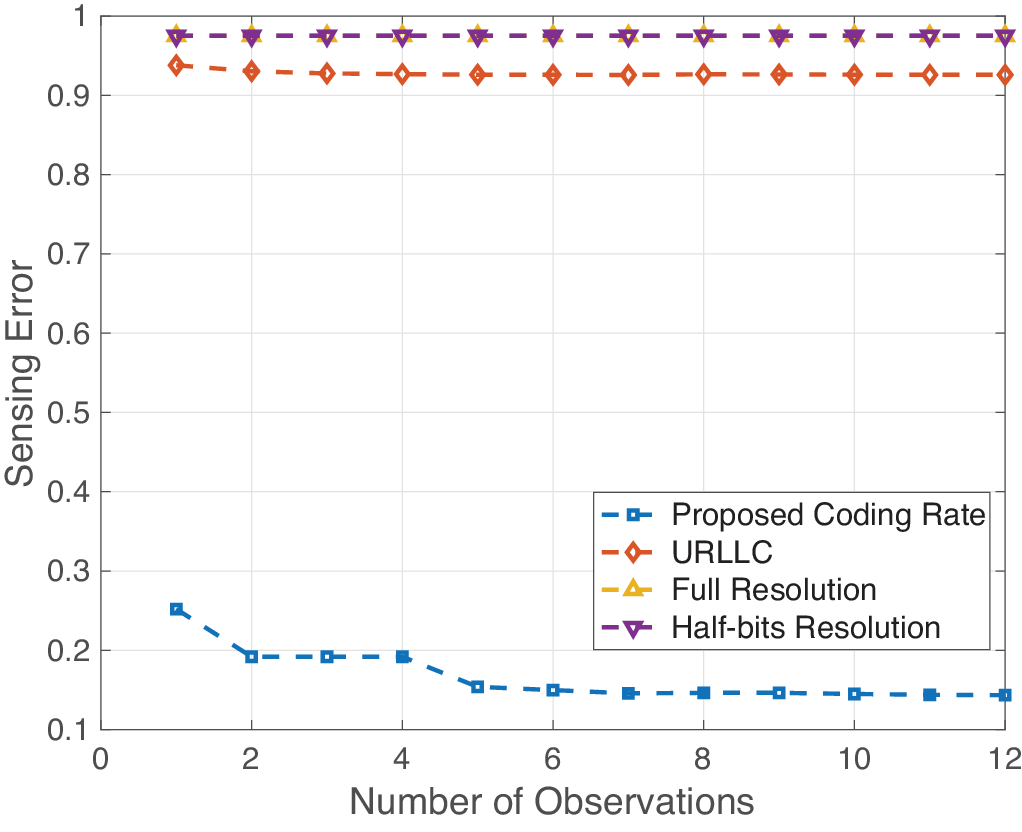}
    \label{CNN Simulation: Effect of observation number}}
    \subfigure[Effect of transmit SNR]{\includegraphics[width=0.48\columnwidth]{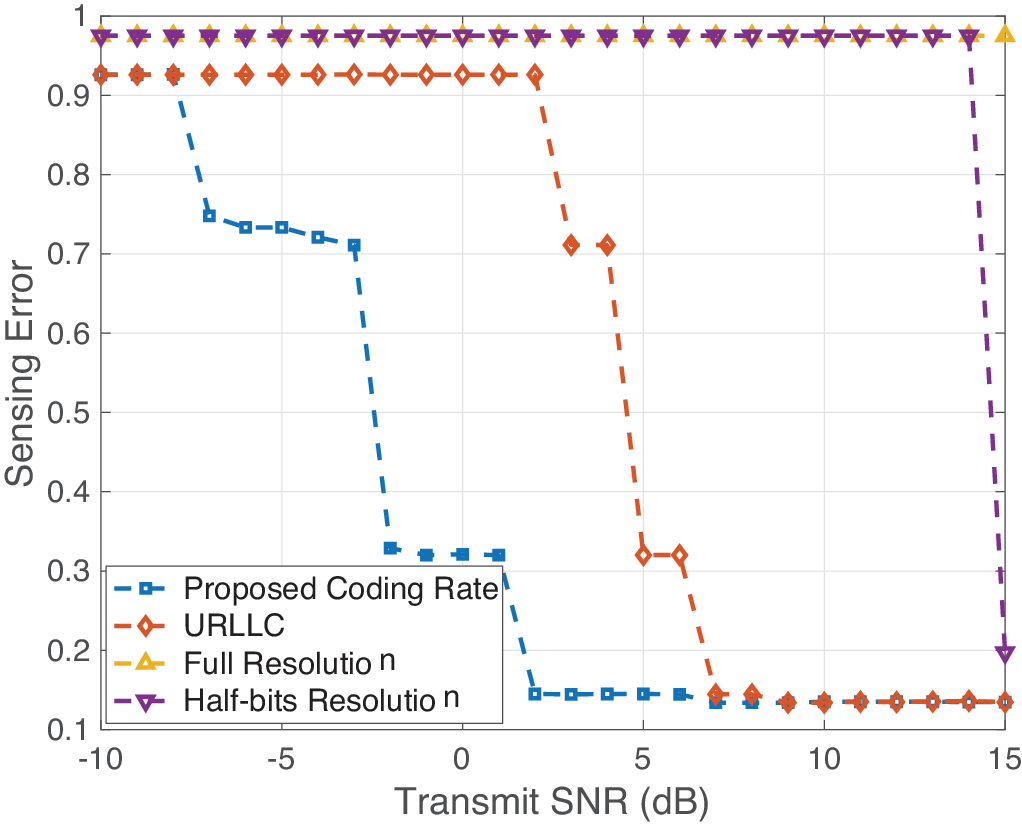}
    \label{CNN Simulation: Effect of transmit SNR}}
    \subfigure[Effect of receive antennas]{\includegraphics[width=0.48\columnwidth]{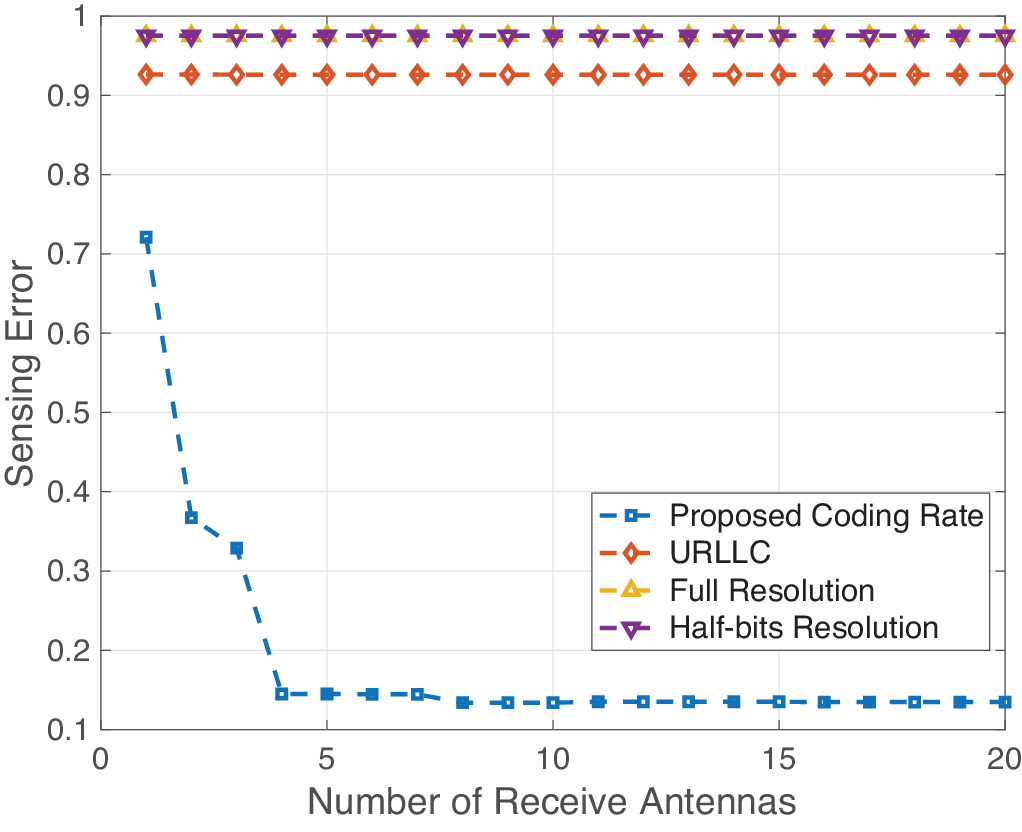}
    \label{CNN Simulation: Effect of receive antennas}}
    \subfigure[Effect of blocklength]{\includegraphics[width=0.48\columnwidth]{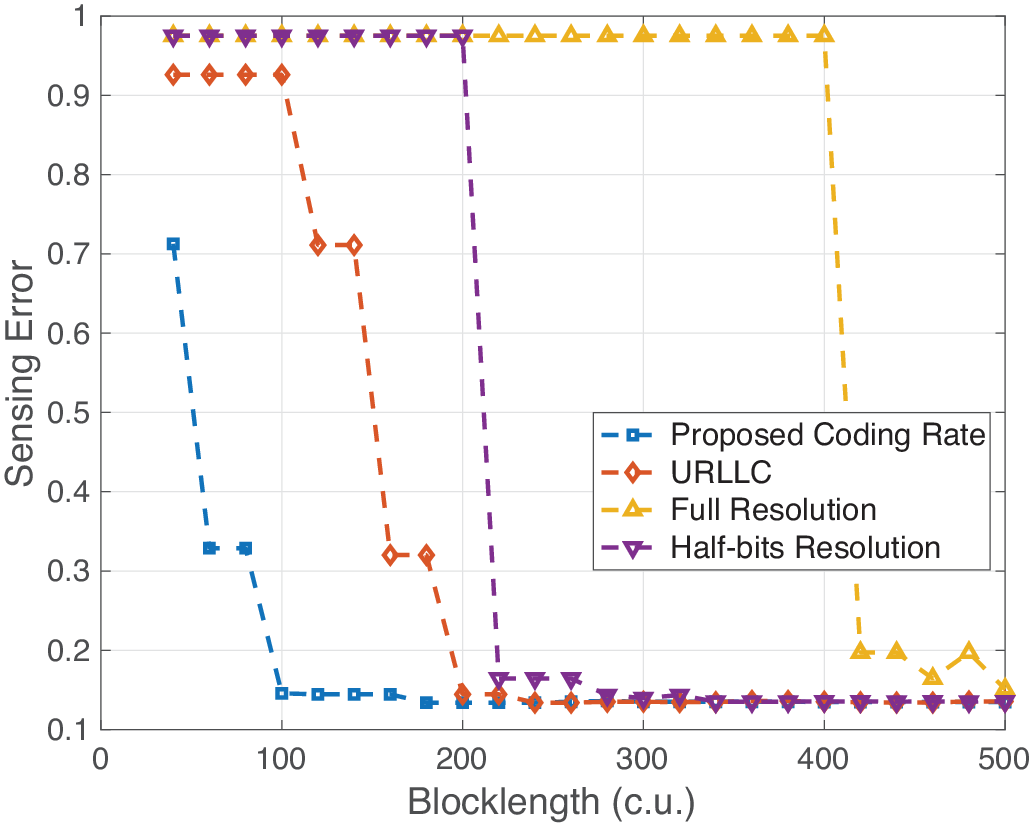}
    \label{CNN Simulation: Effect of blocklength}}
    \caption{The effects of different parameters on sensing error and comparison with benchmarks in the case of CNN classification.\vspace{-5mm}}\label{Simulation: CNN classification}
\end{figure}

\vspace{-6mm}
\subsection{CNN Classification Case}\label{Sec: CNN Classification Case}
We now apply the proposed framework to perform multi-view CNN-based classification on a real dataset. The coding rate is first optimized by leveraging insights derived from statistical inference, and the results are then compared against benchmarks to evaluate performance.
\subsubsection{Coding Rate Optimization on Real Dataset}
To optimize the coding rate for the multi-view CNN system, we address the problem formulated in \eqref{Eqn: objective}. Unlike the statistical inference case, there is no theoretical solution for this scenario due to the unknown discriminant gain of the CNN classifier. To overcome this challenge, we propose a heuristic coding rate optimization approach that incorporates insights from statistical inference.
Specifically, experimental results provide the relationship between inference accuracy, $a$, and the number of bits per feature, $R$. 
Drawing from the statistical inference case, where the tradeoff between source distortion and channel reliability is expressed as a sum of logarithmic terms, we adopt a similar approach here. 
We fit the experimentally obtained relationship between $a$ and $R$ using a logarithmic model, giving that $\ln a(R)=-10R^{-3}-0.2$. 
Based on this fit, we define the surrogate function for the CNN case as
$\phi_{\sf cnn}(R_c) = \ln a(R_c) + \ln(1 - \ep(R_c))$,
where $\ep(R_c)$ represents the average packet loss probability. 
Using this surrogate function, the coding rate can be efficiently optimized via the proposed algorithm.

\subsubsection{Comparison with Benchmarks}
Fig.~\ref{Simulation: CNN classification} presents the performance of the proposed ultra-low-latency EI-Sense system compared to benchmark methods in terms of sensing error. 
The results are evaluated with respect to the number of observations in Fig.~\ref{CNN Simulation: Effect of observation number}, transmit SNR in Fig.~\ref{CNN Simulation: Effect of transmit SNR}, the number of receive antennas in Fig.~\ref{CNN Simulation: Effect of receive antennas}, and blocklength in Fig.~\ref{CNN Simulation: Effect of blocklength}. 
As discussed in the context of statistical inference, similar trends are observed in the multi-view CNN classifier. Specifically, the sensing error probability consistently decreases with an increase in the number of observations, higher transmit SNR, more receive antennas, and longer blocklength.
Our proposed coding rate adaptation outperforms the benchmarks across all settings by balancing the tradeoff between source distortion and channel reliability. A similar trend is observed where other benchmark approaches start to match the performance of our proposed scheme when the transmit SNR is high and the blocklength is long. This is because high transmit SNR and long blocklength enable the transmission of high-resolution feature vectors with high reliability. In the context of the multi-view CNN classifier, our scheme is particularly useful in scenarios characterized by low SNR and short packet lengths, where the benefits of coding rate adaptation for balancing source distortion and channel reliability are most pronounced.

\vspace{-1mm}
\section{Concluding Remarks}\label{Sec: Conclusion}
In this paper, we explored the framework of ultra-low-latency EI-Sense system supported by 6G networks, focusing on the critical trade-off between source distortion and channel reliability. 
To address this, we developed a coding rate optimization scheme to minimize sensing errors by effectively balancing the impacts of source distortion and channel reliability. 
Experimental validation on both synthetic and real datasets demonstrated significant performance improvements compared to traditional URLLC techniques.
Future research will explore advanced quantization and coding strategies to further enhance EI-Sense system performance. 

\vspace{-1mm}    
\appendix
\vspace{-1mm}    
\subsection{Proof of Theorem~\ref{Thm: discriminant gain reduction}}\label{Proof: discriminant gain reduction}
\noindent
\textbf{Lemma A.1.} ``If $\bA$ and $\bB$ are two invertible matrices, then $\bA^{-1}-(\bA+\bB)^{-1}=\bA^{-1}\bB(\bA+\bB)^{-1}$.''

\noindent
\textbf{Lemma A.2.} ``If $\bA$ and $\bB$ are two symmetric positive semi-definite matrices, then $\tr(\bA\bB)\leq\tr(\bA)\tr(\bB)$.''

From the definitions of discriminant gain, we have
\begin{align}
    D_0&=\frac{1}{2}(\bmu_2-\bmu_1)^{\sT}\bSigma^{-1}(\bmu_2-\bmu_1)\nn\\
    &=\frac{1}{2}\tr\{\bSigma^{-1}(\bmu_2-\bmu_1)(\bmu_2-\bmu_1)^{\sT}\}\\
    D(\sigma_q^2)&=\frac{1}{2}(\bmu_2-\bmu_1)^{\sT}(\bSigma+\sigma_q^2\bI)^{-1}(\bmu_2-\bmu_1)\nn\\
    &=\frac{1}{2}\tr\{(\bSigma+\sigma_q^2\bI)^{-1}(\bmu_2-\bmu_1)(\bmu_2-\bmu_1)^{\sT}\}
\end{align}
The differences between them can be upper bounded by
\begin{align}
    &D_0-D(\sigma_q^2)\nn\\
    &=\tr\{\sigma_q^2\bSigma^{-1}(\bSigma+\sigma_q^2\bI)^{-1}(\bmu_2-\bmu_1)(\bmu_2-\bmu_1)^{\sT}\}\nn\\
    &=\tr\{\sigma_q^2(\bSigma+\sigma_q^2\bI)^{-1}(\bmu_2-\bmu_1)(\bmu_2-\bmu_1)^{\sT}\bSigma^{-1}\}\label{Eqn: Appdx A1}\\
    &\leq\tr\{\sigma_q^2(\bSigma+\sigma_q^2\bI)^{-1}\}\tr\{(\bmu_2-\bmu_1)(\bmu_2-\bmu_1)^{\sT}\bSigma^{-1}\}\nn\\
    &=\sigma_q^2\tr\{(\bSigma+\sigma_q^2\bI)^{-1}\}D_0.
\end{align}
Following similar procedures until \eqref{Eqn: Appdx A1}, the lower bound can be derived as 
\begin{align}
    &D_0-D(\sigma_q^2)\nn\\
    &=\frac{\sigma_q^2\tr\{(\bSigma+\sigma_q^2\bI)^{-1}(\bmu_2-\bmu_1)(\bmu_2-\bmu_1)^{\sT}\bSigma^{-1}\}\tr\{\bSigma+\sigma_q^2\bI\}}{\tr\{\bSigma+\sigma_q^2\bI\}}\nn\\
    &\geq\sigma_q^2\tr\{\bSigma+\sigma_q^2\}^{-1}D_0.
\end{align}
This completes the proof.

\subsection{Proof of Theorem~\ref{Thm: performance of sequential sensing}}\label{Proof: performance of sequential sensing}
Consider $M$ out of total $K$ transmissions are successful, where $M$ is a random variable with $M\leq K$. Given the set of successfully received observations $\{\bx\}_M\triangleq\{\bx_1,\cdots,\bx_M\}$, the score function, $\Xi(\{\bx\}_M)$, is the sum of individual discrimination score functions, $\xi(\bx_k),k=1,\cdots,M$, as follows: $\Xi(\{\bx\}_M)=\sum_{k=1}^M\xi(\bx_k)$. If the observations are sampled from the class-1, then the sensing error probability can be upper bounded by
\begin{align}
    P_{e,1}&=\Pr(\Xi(\{\bx\}_M)>0|\ell=1)\\
    &=\Pr(\exp(t\Xi(\{\bx\}_M))>1|\ell=1),~\forall t>0\\
    &\leq\min_{t>0}\mE[\exp(t\Xi(\{\bx\}_M))|\ell=1]\label{Eqn: Appdx B1}\\
    &=\min_{t>0}\mE\l[\mE\l[\exp\l(t\sum_{k=1}^{M}\xi(\bx_k)\r)\Bigg|M\r]\r]\\
    &=\min_{t>0}\mE\l[\prod_{k=1}^M\mE\l[\exp(t\xi(\bx))\Bigg|\ell=1\r]\r]\\
    &=\min_{t>0}\mE\l[\prod_{k=1}^M\cM_{\xi}(t)\r]=\min_{t>0}\mE\l[(\cM_{\xi}(t))^M\r]\label{Eqn: Appdx B7}\\
    &\equiv \min_{t>0}\cG_{M}(\cM_{\xi}(t))
\end{align}
where \eqref{Eqn: Appdx B1} comes from the Markov inequality, $\cM_{\xi}(\cdot)$ is the \emph{moment-generating function} (MGF) of the random variable $\xi(\bx)$ condition on $(\ell=1)$, and $\cG_M(\cdot)$ is the \emph{probability-generating function} (PGF) of the random variable $M$.

Since the observations $\{\bx\}_M$ are sampled from class-1, the discriminant functions follow the Gaussian distribution: $\xi(\bx_k)\sim\cN(-D,2D), k=1,\cdots,M$, whose MGF is 
\begin{align}
    \cM_{\xi}(t)=\exp\l(-Dt+\frac{1}{2}\cdot2Dt^2\r)=\exp(D(t^2-t)).\label{Eqn: Appdx B3}
\end{align}
Due to channel fading, in the total $K$ rounds, the probability of packet loss varies as $\varepsilon_{p,k}, k=1,\cdots,K$. The number of successful transmissions follows the Poisson Binomial distribution: $M\sim{\sf PoiBin}(K;1-\varepsilon_{p,1},\cdots,1-\varepsilon_{p,K})$, whose PGF can be written as
\begin{align}
    \cG_M(z)\!=\!\prod_{k=1}^K(\varepsilon_{p,k}\!+\!(1\!-\!\varepsilon_{p,k})z)\!=\!\prod_{k=1}^K(z\!+\!(1\!-\!z)\varepsilon_{p,k})).
\end{align}
Then, we perform some transformations as follows:
\begin{align}
    \cG_M(z)&=\exp\l[\ln\prod_{k=1}^K(z+(1-z)\varepsilon_{p,k})\r]\\
    &=\exp\l[\sum_{k=1}^K\ln\l(z\l(1+\frac{1-z}{z}\varepsilon_{p,k}\r)\r)\r].
\end{align}

Based on the above results, we can derive that
\begin{align}
    \cG_M(z)
    &=\exp\l[K\ln z+\sum_{k=1}^K\ln\l(1+\frac{1-z}{z}\varepsilon_{p,k}\r)\r]\\
    &=\l[z\exp\l(\frac{1}{K}\sum_{k=1}^K\ln\l(1+\frac{1-z}{z}\varepsilon_{p,k}\r)\r)\r]^K\\
    &\leq \l[z\cdot\frac{1}{K}\sum_{k=1}^K\l(1+\frac{1-z}{z}\varepsilon_{p,k}\r)\r]^K\label{Eqn: Appdx B2}\\
    &=\l(z+(1-z)\overline{\varepsilon_{p}}\r)^K,\label{Eqn: Appdx B4}
\end{align}
where \eqref{Eqn: Appdx B2} comes from Jensen's inequality, and $\overline{\varepsilon_{p}}$ denotes the averaged packet loss probability given by $\overline{\varepsilon_p}=\frac{1}{K}\sum_{k=1}^K\varepsilon_{p,k}$.

The sensing error probability can be obtained by substituting the MGF in \eqref{Eqn: Appdx B3} into the PGF in \eqref{Eqn: Appdx B4}, giving that
\begin{align}
    P_{e,1}&=\min_{t>0}\l(\cM_{\xi}(t)+(1-\cM_{\xi}(t))\overline{\varepsilon_{p}}\r)^K\label{Eqn: Appdx B6}\\
    &=\l(\exp(-D/4)+(1-\exp(-D/4))\overline{\varepsilon_{p}}\r)^K,\label{Eqn: Appdx B5}
\end{align}
where the minimum value of \eqref{Eqn: Appdx B6} is achieved at $t=\frac{1}{2}$. According to the symmetry, the Bayes error probability for observations that sampled from class-2 has the same expression as \eqref{Eqn: Appdx B5}. 
This completes the proof.

\subsection{Proof of Proposition~\ref{Proposition: average packet loss}}\label{Proof: average packet loss}
To facilitate our derivation, we define a function as follows:
\begin{equation}
    g(\gamma)=\frac{1}{2V(\gamma)}\l(C(\gamma)-R_c\r)^2,
\end{equation}
whose minimum is achieved at $\gamma=2^{R_c}-1$.
Then, the integral in \eqref{Eqn: expected packet loss probability} can be calculated as
\begin{align}
    \overline{\varepsilon_p}&=\int_{0}^{\infty}Q\l(\sqrt{2Ng(\gamma)}\r)p_{\gamma}(\gamma)d\gamma\\
    &=\int_{0}^{\infty}\frac{1}{\sqrt{2\pi}}\exp(-Ng(\gamma))\frac{Ng'(\gamma)}{\sqrt{2Ng(\gamma)}}F_{\gamma}(\gamma)\d\gamma,\label{Eqn: Appdx C1}
\end{align}
where $F_{\gamma}(\cdot)$ denotes the cumulative density function (CDF) of the Erlang distribution $\gamma\sim{\sf Erlang}(L,1/\gamma_0)$, and $g'(\cdot)$ denotes the first derivative of function $g(\cdot)$.
The equation in \eqref{Eqn: Appdx C1} comes from the two facts that $F_{\gamma}(0)=0$ and $\lim_{\gamma\to\infty}Q(\sqrt{2Ng(\gamma)})=0$.
Leveraging the Laplace's method, we can approximate the integral in \eqref{Eqn: Appdx C1} as
\begin{align}
    \overline{\varepsilon_p}&=\frac{1}{2}\sqrt{\frac{N}{\pi}}\int_{0}^{\infty}\frac{g'(\gamma)\l(\Gamma(L)-\Gamma(L,\gamma/\gamma_0)\r)}{\Gamma(L)\sqrt{g(\gamma)}}\exp(-Ng(\gamma))\d\gamma\nn\\
    &\approx\sqrt{\frac{g'(\gamma)^2}{2g''(\gamma)g(\gamma)}}\l(1-\frac{\Gamma(L,\gamma/\gamma_0)}{\Gamma(L)}\r)\exp(-Ng(\gamma))\Bigg|_{\gamma=2^{R_c}-1}\nn\\
    &=1-\frac{1}{\Gamma(L)}\Gamma\l(L,\frac{1}{\gamma_0}\l(2^{R_c}-1\r)\r),\label{Eqn: Appdx C2}
\end{align}
where $\Gamma(\cdot)$ and $\Gamma(\cdot,\cdot)$ denote the Gamma function and upper incomplete Gamma function, respectively. 
The equation in \eqref{Eqn: Appdx C2} comes from the fact that $\frac{g'(\gamma)^2}{2g''(\gamma)g(\gamma)}\Big|_{\gamma=2^{R_c}-1}=1$, which can be more easily verified from the following result:
\begin{equation}
    g'(\gamma)^2-2g''(\gamma)g(\gamma)\big|_{\gamma=2^{R-c}-1}=0.
\end{equation}
This completes the proof.

\subsection{Proof of Lemma~\ref{Lemma: convexity of DG}}\label{Proof: convexity of DG}
Since $\bSigma\in\mR^{d\times d}$ is a symmetric positive definite matrix, we can perform the eigenvalue decomposition as follows:
\begin{equation}
    \bSigma=\bQ^{\sT}\diag(\lambda_1,\cdots,\lambda_{d})\bQ,
\end{equation}
where $\bQ\in\mR^{d\times d}$ is an orthogonal matrix and $\lambda_i>0,~\forall i$.
To ease the notation, we define $\bnu\triangleq\bQ(\bmu_1-\bmu_2)$. 
Then, the effective discriminant gain can be denoted as
\begin{align}
    D(R_c)&=\frac{1}{2}\bnu^{\sT}\diag\l\{\frac{1}{\lambda_1+\sigma_q^2(R_c)},\cdots,\frac{1}{\lambda_d+\sigma_q^2(R_c)}\r\}\bnu\nn\\
    &=\frac{1}{2}\sum_{i=1}^d\frac{\nu_i^2}{\lambda_i+\sigma_q^2(R_c)},
\end{align}
where $\nu_1,\cdots,\nu_d$ are the coefficients in vector $\bnu$.
It is obvious that $z=2^{NR_c/d}-1$ is convex and $\sigma_q^2=\frac{U^2}{3z^2}$ is monotone decreasing for $z>0$, so that $\sigma_q^2(R_c)$ is concave. Therefore, $\frac{\nu_i^2}{\lambda_i+\sigma_q^2(R_c)}$ is concave for $i=1,\cdots,d$.
The sum of them is still concave, i.e., $D(R_c)$ is concave.
This completes the proof.

\subsection{Proof of Lemma~\ref{Lemma: convexity of epsilon}}\label{Proof: convexity of epsilon}
Based on the average packet loss probability in \eqref{Eqn: average packet loss probability}, the average packet success probability can be expressed as 
\begin{align}
    1-\ep(R_c)&=\frac{1}{\Gamma(L)}\Gamma\l(L,\frac{1}{\gamma_0}(2^{R_c}-1)\r)\nn\\
    &=\frac{1}{\Gamma(L)}\int_{\frac{1}{\gamma_0}(2^{R_c}-1)}^{\infty}t^{L-1}e^{-t}\d t.
\end{align}
Transforming the variable as $x=t-\frac{1}{\gamma_0}(2^{R_c}-1)$, we have
\begin{align}
    1-\ep(R_c)=&~\frac{1}{\Gamma(L)}\int_{0}^{\infty}\l(x+\frac{1}{\gamma_0}(2^{R_c}-1)\r)^{L-1}\nn\\
    &\times\exp\l(-\l(x+\frac{1}{\gamma_0}(2^{R_c}-1)\r)\r)\d x.
\end{align}
The derivative of average packet loss probability is 
\begin{align}
    &\ep'(R_c)=-\frac{1}{\Gamma(L)}\frac{\d}{\d R_c}\int_{\frac{1}{\gamma_0}(2^{R_c}-1)}^{\infty}t^{L-1}e^{-t}\d t\nn\\
    &=-\frac{2^{R_c}(2^{R_c}-1)^{L-1}\ln2}{\gamma_0^L\Gamma(L)}\exp\l(-\frac{1}{\gamma_0}(2^{R_c}-1)\r).
\end{align}
To facilitate our derivation, we introduce a function as follows:
\begin{align}
    \zeta(R_c)&\triangleq\frac{1-\ep(R_c)}{\ep'(R_c)}\nn\\
    &=\underbrace{\frac{\gamma_0}{2^{R_c}\ln2}}_{(a)}\int_{0}^{\infty}\underbrace{\l(1+\frac{\gamma_0x}{2^{R_c}-1}\r)^{L-1}}_{(b)}e^{-x}\d x.
\end{align}
Since $\gamma_0>0$ and $L\geq1$, both $(a)$ and $(b)$ are decreasing function with respect to $R_c$. 
Hence, $\zeta(R_c)$ is a decreasing function.
According to the relation that
\begin{equation}
    \frac{\d}{\d R_c}\ln(1-\ep(R_c))=-\frac{1}{\zeta(R_c)},
\end{equation}
we can conclude that the first derivative of $\ln(1-\ep(R_c))$ is a decreasing function with respect to $R_c$. 
This means the function $\ln(1-\ep(R_c))$ is a concave function, and thus $(1-\ep(R_c))$ is log-concave.
This completes the proof.

\bibliography{Reference}
\bibliographystyle{IEEEtran}

\end{document}